\documentclass[fleqn,10pt]{wlscirep}
\usepackage[utf8]{inputenc}
\usepackage[T1]{fontenc}

\usepackage{amsthm}
\usepackage{graphicx} 
\usepackage{algpseudocode,algorithm,algorithmicx}
\usepackage{hyperref}

\usepackage{pifont}
\usepackage{amsmath}

\usepackage{centernot}
\theoremstyle{definition}

\theoremstyle{theorem}
\newtheorem*{theorem}{Theorem}

\theoremstyle{remark}

\title{Fuzzy $k$-anonymity in complex networks}

\author[1,2,*]{Rachel G. de Jong}
\author[2, 1]{Mark P. J. van der Loo}
\author[1]{Frank W. Takes}
\affil[1]{Leiden University, LIACS, 2333 CA Leiden, The Netherlands} 
\affil[2]{Statistics Netherlands, Research and Development, 2492JP The Hague, The Netherlands}
\affil[*]{r.g.de.jong@liacs.leidenuniv.nl}

\begin{abstract}
With the introduction of large-scale network data, including population-scale social networks, techniques for privacy-aware sharing of network data become increasingly important. 
While existing $k$-anonymity approaches can model different attacker scenarios, they typically assume that attacker knowledge exactly matches the published network structure.
We argue that exact knowledge is often unrealistic and introduce $\phi$-$k$-anonymity, a fuzzy variant of $k$-anonymity in which parameter $\phi$ captures the level of uncertainty in attacker knowledge. 
Across a benchmark of 39 real-world networks, a realistic level of uncertainty ($\phi=5\%$) renders, on average, $64\%$ of previously unique nodes anonymous. 
To further enhance anonymity, we apply anonymization algorithms under a 5\% edge modification budget.
While full anonymization is often unattainable under exact $k$-anonymity, with low uncertainty ($\phi=10\%$) our newly proposed \textsc{greedy} algorithm anonymizes over 99\% of the nodes.
Uncertainty also enables effective anonymization in otherwise difficult to anonymize dense synthetic graphs.
Additionally, data utility in terms of structural properties and performance on network analysis tasks is well preserved, with most metrics changing less than 5\%.
Overall, our findings suggest that modest uncertainty assumptions yield high levels of anonymity and utility, motivating further research on uncertainty-aware privacy guarantees for network data.
\end{abstract}

\begin{document}

\flushbottom
\maketitle

\thispagestyle{empty}

\section*{Introduction}\label{sec:intro}
Network science is a growing field with a variety of applications in a wide range of domains.
In social networks this includes, for example, modeling influence~\cite{li2018influence} or epidemic disease spread,~\cite{azizi2020epidemics} and measuring social segregation.~\cite{kazmina2024socio, bojanowski2014measuring}
For these purposes it is desirable to have access to real-world networks describing meaningful relationships between actual people.
A notable recent example is the increasing interest in population-scale networks based on administrative data.~\cite{bokanyi2023anatomy, panayiotou2025anatomy,cremers2025unveiling}
However, such complex networks may contain sensitive information about individuals, and sharing or publishing this network data may lead to a breach in privacy.
Even after pseudonymization, i.e., removing unique identifiers such as names, there is still an identification risk as it is known that entities can still be uniquely identified based on the network structure.~\cite{hay2008resisting, romanini2021privacy, dejong2023effect}

A commonly used approach to measure identification risk, $k$-anonymity, was initially introduced in the field of statistical disclosure control for tabular data describing entities and their attributes.~\cite{duncan_disclosure, lambert1993measures}
Unlike differential privacy, another commonly used approach for sharing data,~\cite{jiang2021applications} $k$-anonymity allows one to share a perturbed anonymized version of the network, while guaranteeing privacy constraints on the node level and hence protecting against identification risk.
When applied to network data, $k$-anonymity can be used to assess whether a given graph is $k$-anonymous.~\cite{hay2008resisting, romanini2021privacy, dejong2023effect}
A network satisfies $k$-anonymity if for each node in the graph there are at least $k-1$ other nodes with the same \textit{signature} according to a specified \textit{measure} describing the structure surrounding a node. 
Commonly used examples of such anonymity measures are the degree of the node~\cite{liu2008towards} or the exact graph structure surrounding a node.~\cite{dejong2023algorithms, zhou2008preserving}
Real-world networks are often not anonymous, and an anonymization algorithm should be applied to perturb the network and ensure that the resulting (anonymized) network satisfies $k$-anonymity.~\cite{de2026anonymization, bonello2025ga, arsene2026simulated}

Each of the anonymity measures models a different scenario where a potential attacker, who tries to obtain sensitive information, is assumed to have access to this type of information.~\cite{dejong2024comparison}
As the choice of measure essentially determines against which attacker scenario one protects, it is important to select it carefully. 
The measure should not be too lenient, or a possible attacker may still be able to identify nodes and extract (sensitive) information. 
At the same time, if the measure is too strict, anonymizing the network typically requires many alterations. As a result, the network structure and performance on network analysis tasks change drastically, making the anonymized network less suitable for further use, i.e., reducing its \emph{data utility}.~\cite{dejong2024comparison} 
Hence, the choice in measure affects the balance between privacy and data utility.
To ensure that the damage to data utility is limited, a budget can be given to the anonymization algorithm, so that it deletes up to a certain fraction of edges, limiting the impact on data utility.~\cite{de2026anonymization}

\begin{figure}[t!]
    \centering
    \includegraphics[width=\textwidth]{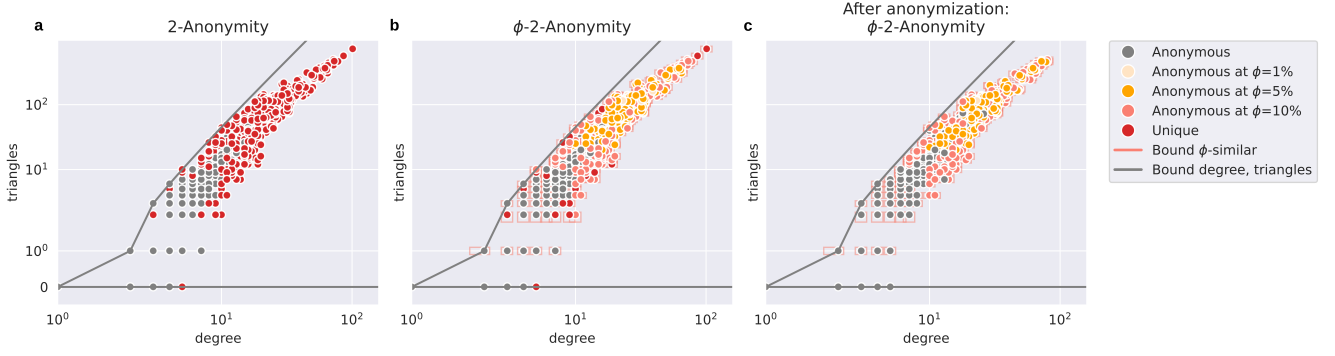}
    \caption{Exact $k$-Anonymity and fuzzy $\phi$-$k$-anonymity in the ``Copnet FB'' network dataset before and after anonymization. 
    \textbf{a:} Each dot represents the node degree (horizontal axis) and the number of triangles the node is part of (vertical axis) for at least one node, together forming the $(deg, tri)$ signature. 
    Grey dots represent 2-anonymous nodes, red dots represent unique nodes.
    The gray lines indicate the minimum and maximum number of triangles for the given degree.
    \textbf{b:} Fuzzy setting. Yellow, orange and pink nodes are $\phi$-$2$-anonymous for the corresponding $\phi$ values. 
    Rectangles indicate the area to which each center nodes is $\phi$-similar with $\phi=10\%$. 
    \textbf{c:} Node values and $\phi$-$2$-anonymity after applying \textsc{greedy} anonymization with a 5\% budget.
     }
    \label{fig:fuzzy-example}
\end{figure}

A key property of the measures introduced for $k$-anonymity so far in the literature is the assumption that a possible attacker has \textit{exact information}.
For example, an attacker may know the exact degree of a node or its exact surrounding graph structure, precisely matching the information in the considered network.
Hence, the $k$-anonymity approach labels two nodes as \textit{equivalent} or \textit{not equivalent}.
Although this assumption of exact information models a worst-case scenario and with that gives an upper bound on the risk of publishing network data, this assumption is not always realistic.
In particular, in many scenarios it is in fact very likely that the attacker does not have exact information; even in real-world data collection it is common to have noisy or, by the time analysis takes place, outdated information. 
For example, consider a pseudonymized Facebook friendship network of a city.
Even when an attacker relies on the degree of the node to identify individuals, substantial uncertainty remains.
First, Facebook data are highly dynamic: users add and remove connections, and accounts appear and disappear over time.
Hence, data from a snapshot taken at a certain time likely differs from the network of today.
There may also be uncertainty about who is included in the dataset: users may not update their location, or maintain fake accounts. 
Uncertainty is amplified when attackers rely on external sources rather than direct platform access, as the amount of information that can be obtained about a person using open data sources varies substantially.~\cite{deVries2023risk}

In this work, to account for different levels of uncertainty in the attacker scenario, we introduce $\phi$-$k$-anonymity, a fuzzy notion of $k$-anonymity.
This notion requires both an anonymity measure to determine equivalence of nodes and a value $\phi$, which acts as a proxy for the level of uncertainty.
To operationalize uncertainty in a computationally tractable yet expressive way, we adopt the $(deg, tri)$ anonymity measure.
This measure is similar to a measure introduced in previous work,~\cite{dejong2024comparison} where node signatures consist of the number of connections and edges in its direct neighborhood, and is directly interpretable in scenarios with attacker uncertainty.
The $(deg, tri)$ measure determines for each node the degree $deg$ and the number of triangles $tri$ it is part of.
This measure is positioned in between the \textsc{degree} measure, which yields very low anonymity and represents a very weak attacker scenario~\cite{dejong2024comparison} and \textsc{$d$-$k$-anonymity},~\cite{dejong2023algorithms} a strict measure assuming complete knowledge of a node's surrounding structure.
For the \textsc{degree} measure, anonymization with as few alterations as possible can be done in polynomial time, while no efficient exact algorithms exist for the more complex measures.~\cite{dejong2024comparison}
Computing similarity between nodes for the latter measure would require the NP-hard to compute graph edit distance.~\cite{blumenthal2020exact}
Notably, it was found that in real-world networks $(deg, tri)$ often measures anonymity levels comparable to \textsc{$d$-$k$-anonymity},~\cite{dejong2024comparison} making it a promising choice for modeling uncertainty while avoiding computational bottlenecks of more complex measures.
A more elaborate discussion of the chosen measure's positioning in the literature and further rationale is provided in Supplementary Information.

Given the notion of $\phi$-$k$-anonymity under the $(deg, tri)$ measure, a node $v$ is $\phi$-similar to a node $w$ if the relative difference in both  degree ($\Delta deg/deg$) and number of triangles ($\Delta tri / tri$) are at most $\phi$. 
Hence, the differences allowed for a node to be $\phi$-similar to another are relative to the node's values.
This reflects the scenario in which an attacker is more certain about the structure of nodes with small neighborhoods, where a few additional edges result in a relatively large difference.
Conversely, we assume that an attacker is less certain about nodes with large neighborhoods, for which missing or extra edges have a smaller relative impact. 
As a result, if a node $v$ is $\phi$-similar to a node $w$, it does not necessarily follow that $w$ is $\phi$-similar to $v$, making $\phi$-similarity a non-symmetric relation.
This behavior is illustrated by the rectangles in Fig.~\ref{fig:fuzzy-example}\textbf{b} and \textbf{c}.

\begin{figure}[t!]
    \centering
    \includegraphics[width=0.85\textwidth]{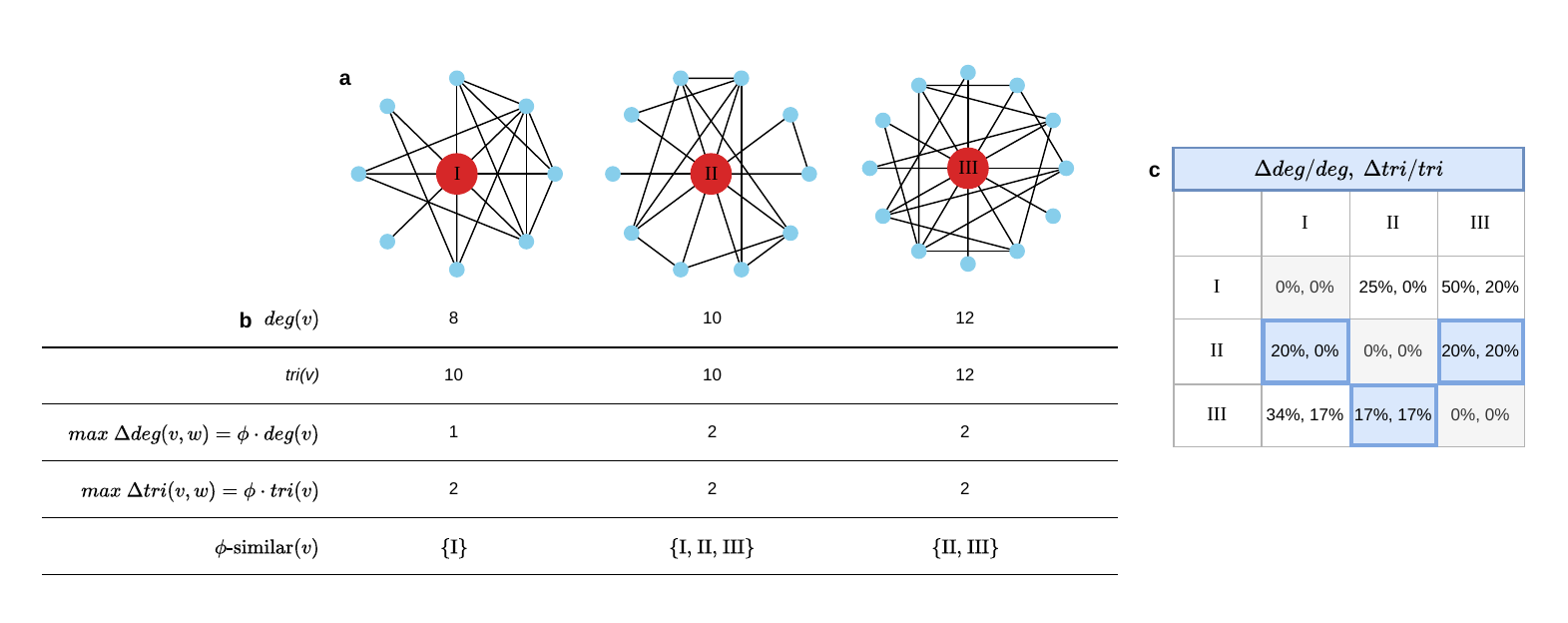}
    \caption{Example illustrating $\phi$-$k$-anonymity with $\phi=20\%$. \textbf{a:} Ego network of the red nodes labeled I, II and III. \textbf{b:} Table with each node's degree $deg(v)$, triangle count $tri(v)$, maximum allowed difference in $deg$ and $tri$ and the set of nodes to which it is $\phi$-similar. 
    \textbf{c:} Relative difference in $deg$ and $tri$ values from one node (row) to another node (column). Blue cells indicate that $v$ (row) is $\phi$-similar to $w$ (column). Grey cells (0\% difference) indicate identical signatures.}
    \label{fig:example}
\end{figure}

We demonstrate the workings of $\phi$-$k$-anonymity with $\phi=20\%$ using the example in Fig.~\ref{fig:example}.
As illustrated in Fig.~\ref{fig:example}\textbf{b}, for node II, with degree 10 ($deg=10)$ and 10 triangles ($tri=10$),
$\phi$-similar nodes can differ by at most 2 in both $deg$ and $tri$.
Figure~\ref{fig:example}\textbf{c} shows the percentage difference $\Delta deg,\ \Delta tri$ to other nodes.
If a node differs by no more than $\phi=20\%$ in both degree and number of triangles, it is $\phi$-similar to the considered node. 
Using this criterion, node II is $\phi$-similar to nodes I, III, and itself.
At the same time, node I is not $\phi$-similar to node II because the degrees of all other nodes differ by more than $20\%$. 
Moreover, node I is not considered $\phi$-similar to any other node in the example  as all degree differences exceed the threshold, again showing that $\phi$-similarity is not a symmetric relation.

In the experimental part of this paper, we investigate how varying levels of uncertainty in the attacker's knowledge affect both identification risk and the performance of anonymization algorithms under a modification budget of 5\% of the edges.
We evaluate our framework for fuzzy $k$-anonymity on both synthetic graph models and a benchmark dataset with 39 real-world networks from diverse domains and with varying structural properties.
In addition, we assess how well utility is preserved under different uncertainty levels and compare these outcomes to the standard non-fuzzy setting where $\phi=0\%$. 

Our results demonstrate that even small levels of uncertainty substantially increase anonymity.
In graph models with varying densities, allowing $\phi=5\%$ renders the majority of nodes anonymous.
Across the 39 real-world networks, introducing $\phi=5\%$ makes, on average, 64\% of the nodes that were unique under exact knowledge, anonymous.
Uncertainty also substantially improves the effectiveness of anonymization algorithms: even settings in which dense synthetic graphs cannot be anonymized under $\phi=0\%$ become tractable once uncertainty is incorporated.
Moreover, with just $\phi=10\%$, the \textsc{greedy} algorithm anonymizes nearly all nodes in the real-world benchmark networks.

Overall, our findings suggest that when modest deviations in attacker knowledge are realistic, identification risks may be substantially overestimated by non-fuzzy measures.
Allowing limited structural variation therefore makes it attainable to close to fully anonymize networks while preserving high utility.
Together, these results challenge the default assumption of exact attacker knowledge and motivate further research on uncertainty-aware privacy guarantees for network data.

The remainder of this paper is structured as follows.
The~\hyperref[sec:results]{Results} section presents experimental results on synthetic and real-world networks using $\phi$-$k$-anonymity.
We conclude our paper in the~\hyperref[sec:discussion]{Discussion} section with a summary and directions for future work.
Lastly, the~\hyperref[sec:methods]{Methods} section contains a more elaborate description of methods and experimental setup used.

\section*{Results}\label{sec:results}
In this section, we use $\phi$-$k$-anonymity to evaluate how accounting for uncertainty in attacker knowledge affects network anonymity and the performance of anonymization algorithms on synthetic graph models and real-world networks.
We consider three synthetic graph models: (1) the Erdős–Rényi model (ER),~\cite{erdos1960evolution} in which each edge exists independently with equal probability, (2) the Barabási–Albert model (BA),~\cite{barabasi1999emergence} which produces graphs with a power-law degree distribution, and (3) the Watts–Strogatz model (WS),~\cite{watts1998collective} with rewiring probability $0.05$.
The real-world network datasets used, along with descriptions of their structural properties, are listed in Table~\ref{tab:data}.
The chosen networks span social, human communication and collaboration networks from varying sources and with diverse sizes and topological properties.

\begin{table}[b!]
\scriptsize
\centering
\begin{tabular}{@{}rrrrrrr@{}}
\toprule
Network             & Nr. nodes & Nr. edges & \begin{tabular}[c]{@{}r@{}}Average\\ clustering\\coefficient\end{tabular} & \begin{tabular}[c]{@{}r@{}}Average \\ shortest\\ path\\ length\end{tabular} & \begin{tabular}[c]{@{}r@{}}Frac. \\ nodes in\\ LCC\end{tabular} & Modularity \\ \midrule
Radoslaw emails~\cite{kunegis2013konect}       & 167       & 3,250     & 0.69 & 1.97  & 1.00 & 0.13 \\
Moreno innov.~\cite{kunegis2013konect}       & 241       & 923       & 0.31 & 2.47  & 0.49 & 0.69 \\
Primary school~\cite{stehle2011high}      & 242       & 8,317     & 0.53 & 1.73  & 1.00 & 0.28 \\
Gene fusion~\cite{kunegis2013konect}         & 291       & 279       & 0.00 & 3.90  & 0.38 & 0.87 \\
Copnet calls~\cite{sapiezynski2019copenhagen}        & 536       & 621       & 0.25 & 7.37  & 0.65 & 0.87 \\
Copnet sms~\cite{sapiezynski2019copenhagen}          & 568       & 697       & 0.22 & 7.32  & 0.80 & 0.84 \\
Copnet FB~\cite{sapiezynski2019copenhagen}           & 800       & 6,418     & 0.32 & 2.98  & 1.00 & 0.47 \\
FB Reed98~\cite{networksrepository}           & 962       & 18,812    & 0.33 & 2.46  & 1.00 & 0.33 \\
Arenas email~\cite{kunegis2013konect}        & 1,133     & 5,451     & 0.25 & 3.61  & 1.00 & 0.58 \\
Euroroads~\cite{kunegis2013konect}            & 1,174     & 1,417     & 0.02 & 18.37 & 0.89 & 0.88 \\
Air traffic control~\cite{kunegis2013konect} & 1,226     & 2,408     & 0.07 & 5.93  & 1.00 & 0.71 \\
Network science~\cite{kunegis2013konect}     & 1,461     & 2,742     & 0.88 & 5.82  & 0.26 & 0.96 \\
FB Simmons81~\cite{networksrepository}        & 1,518     & 32,988    & 0.33 & 2.57  & 0.99 & 0.48 \\
DNC emails~\cite{kunegis2013konect}           & 1,866     & 4,384     & 0.59 & 3.37  & 0.98 & 0.59 \\
Moreno health~\cite{kunegis2013konect}       & 2,539     & 10,455    & 0.15 & 4.56  & 1.00 & 0.64 \\
FB Wellesley22~\cite{networksrepository}      & 2,970     & 94,899    & 0.27 & 2.59  & 1.00 & 0.38 \\
Bitcoin alpha~\cite{networksrepository}       & 3,783     & 14,124    & 0.28 & 3.57  & 1.00 & 0.48 \\
US power grid~\cite{kunegis2013konect}        & 4,941     & 6,594     & 0.11 & 18.99 & 1.00 & 0.94 \\
GRQC collab.~\cite{snapnets}         & 5,241     & 14,484    & 0.69 & 6.05  & 0.79 & 0.87 \\
FB Carnegie49~\cite{networksrepository}       & 6,637     & 249,967   & 0.29 & 2.74  & 1.00 & 0.43 \\
Pajek Erdős~\cite{kunegis2013konect}         & 6,927     & 11,850    & 0.40 & 3.78  & 1.00 & 0.71 \\
ChG-miner~\cite{biosnapnets}            & 7,341     & 15,138    & 0.00 & 6.15  & 0.90 & 0.76 \\
DG assoc.~\cite{biosnapnets}           & 7,813     & 21,357    & 0.00 & 4.23  & 1.00 & 0.53 \\
FB GWU54~\cite{networksrepository}            & 12,193    & 469,528   & 0.22 & 2.83  & 1.00 & 0.46 \\
Anybeat~\cite{networksrepository}              & 12,645    & 49,132    & 0.40 & 3.17  & 1.00 & 0.44 \\
CE-CX~\cite{networksrepository}               & 15,229    & 245,952   & 0.23 & 3.85  & 0.99 & 0.62 \\
Astrophysics~\cite{networksrepository}        & 18,771    & 198,050   & 0.68 & 4.19  & 0.95 & 0.64 \\
FB BU10~\cite{networksrepository}              & 19,700    & 637,528   & 0.20 & 3.03  & 1.00 & 0.47 \\
FB Uillinois20~\cite{networksrepository}       & 30,809    & 1,264,428 & 0.22 & 2.99  & 1.00 & 0.48 \\
Enron email~\cite{snapnets}          & 36,692    & 183,831   & 0.72 & 4.03  & 0.92 & 0.62 \\
FB Penn94~\cite{networksrepository}            & 41,536    & 1,362,220 & 0.22 & 3.12  & 1.00 & 0.49 \\
FB wall 2009~\cite{kunegis2013konect}         & 45,813    & 183,412   & 0.15 & 5.60  & 0.96 & 0.66 \\
Brightkite~\cite{networksrepository}          & 58,228    & 214,078   & 0.27 & 4.92  & 0.97 & 0.69 \\
The marker cafe~\cite{dataforgoodlab}     & 69,413    & 1,644,843 & 0.24 & 3.06  & 1.00 & 0.29 \\
Slashdot zoo~\cite{kunegis2013konect}         & 79,116    & 467,731   & 0.09 & 4.04  & 1.00 & 0.35 \\
Twitter~\cite{kunegis2013konect}             & 465,017   & 833,540   & 0.06 & 4.59  & 1.00 & 0.69 \\
DBLP~\cite{kunegis2013konect}                    & 1,824,701 & 8,344,615 & 0.73 & 5.74  & 0.91 & 0.76 \\
Flixster~\cite{kunegis2013konect}             & 2,523,386 & 7,918,801 & 0.21 & 4.82  & 1.00 & 0.62 \\
Youtube~\cite{kunegis2013konect}              & 3,223,585 & 9,375,374 & 0.17 & 5.29  & 1.00 & 0.67  \\ \bottomrule
\end{tabular}
\caption{Real-world network data used in experiments. From left to right: the number of nodes and edges, average clustering coefficient, average shortest path length, fraction of nodes in the largest connected component (LCC) and modularity, an indicator of the quality of the community structure found by the Leiden~\cite{traag2019louvain} community detection algorithm.
}
\label{tab:data}
\end{table}

Similar to previous work,~\cite{romanini2021privacy, dejong2023effect} we focus on the case where $k=2$, unless otherwise specified.
This allows us to summarize anonymity for a graph as a whole as the fraction of $\phi$-$2$-anonymous nodes or uniqueness, the fraction of nodes that is unique or not $\phi$-$2$-anonymous. 
A node is $\phi$-$k$-anonymous if it is $\phi$-similar to at least $k-1$ other nodes, as defined in Equation~\ref{eq:phisim}.
Here $G = (V, E)$ denotes the network, $V$ the set of its nodes and $\phi \in [0, 1]$, often indicated as a percentage, the level of allowed uncertainty. 
$deg(v)$ and $tri(v)$ denote the degree and number of triangles incident to a node $v$ and $\Delta deg(v, w)$, $\Delta tri(v, w)$ denote the absolute difference in degree and number of triangles of node $v$ and $w$.
\begin{equation}
    \phi\text{-}similar(v) = \{w \in V : \Delta deg(v, w)/deg(v) \leq \phi \ \wedge \ \Delta tri(v, w)/tri(v) \leq \phi \}
    \label{eq:phisim}
\end{equation}

For anonymization, we use two algorithms from recent literature and one newly proposed algorithm, with the aim to assess how accounting for uncertainty improves the anonymization process.
The first algorithm, Edge Sampling (\textsc{es}), serves as a baseline algorithm as it removes edges at random.~\cite{romanini2021privacy}
Second, Unique Affected (\textsc{ua}) is a heuristic algorithm that is more likely to remove edges that affect many unique nodes, i.e., edges whose removal causes nodes to change their signature.
\textsc{ua} is computationally more expensive than \textsc{es}, but has been shown to be more effective.~\cite{de2026anonymization}
Third, we introduce the \textsc{greedy} algorithm, which iteratively removes the edges that improve overall network anonymity the most.
While \textsc{ua} heuristically selects edges likely to improve anonymity, \textsc{greedy}, though even more computationally expensive, guarantees the best choice at each step.
To account for utility of the resulting anonymized networks, we consider a budgeted setting~\cite{de2026anonymization} in which we delete at most 5\% of the edges in the network.
Further details on the algorithms and experimental setup are provided in the~\hyperref[sec:methods]{Methods} section.

\begin{figure}[t!]
    \centering
    \includegraphics[width=\textwidth]{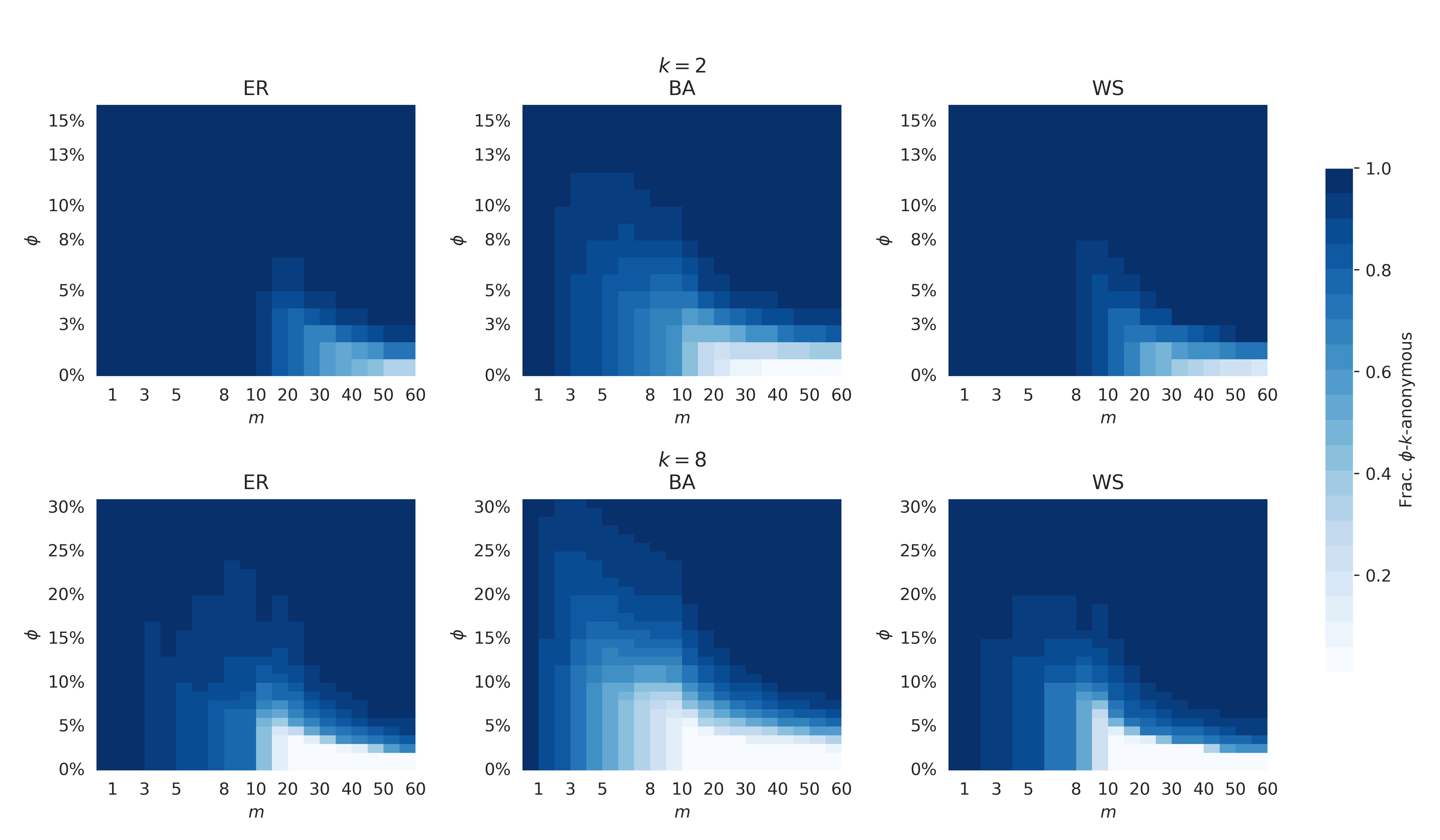}
    \caption{Fuzzy $\phi$-$k$-anonymity for $k=2$ (top) and $k=8$ (bottom) in Erdős–Rényi (ER, left) Barabási–Albert (BA, middle) and Watts Strogatz (WS, right) graph models with 500 nodes.
    The horizontal axis in each subfigure denotes the number of edges per node $m$, the vertical axis the level of uncertainty $\phi$.
    Color indicates the fraction of $\phi$-$k$-anonymous nodes ranging from white, all nodes unique, to dark blue, all nodes are anonymous.}
    \label{fig:fuzzy-models}
\end{figure}

\subsection*{Fuzzy $k$-anonymity in graph models}\label{sub:fuzzymodel}
Figure~\ref{fig:fuzzy-models} shows the fraction of $\phi$-$2$-anonymous nodes (indicated by color) under $\phi$-$k$-anonymity for $k=2$ and $k=8$ across a range of $\phi$ values in three synthetic graph models. 
The bottom row of the matrix in each subfigure, corresponding to $\phi=0\%$, shows, consistent with previous work in the non-fuzzy setting, that the fraction of unique nodes transitions from zero (all nodes anonymous) to one (no nodes anonymous) as the average degree $m$, which serves as a proxy for network density, increases.~\cite{dejong2023effect, romanini2021privacy}
As the average degree increases, neighborhoods become more distinct, which generally increases node uniqueness and hence reduces anonymity. 
At the same time, higher density can also improve $\phi$-$k$-anonymity: as $\phi$ allows differences relative to the node and edge count, higher-degree nodes with more triangles can tolerate larger deviations while still being considered $\phi$-similar to another node.
To show how results generalize to larger synthetic graphs with 1,000 nodes, additional results are included in Supplementary Information.
These results suggest that similar trends hold for larger synthetic graphs, which achieve higher overall anonymity and require less uncertainty to render most nodes anonymous.

The three top subfigures of Fig.~\ref{fig:fuzzy-models}, corresponding to $k=2$, show that unique nodes gradually become anonymous (blue) as $\phi$ increases, in all graph models. 
Even relatively low uncertainty ($\phi=5\%$) is sufficient to make most nodes anonymous for all models and density values.
This suggests that many nodes that are unique when $\phi=0\%$ are still structurally similar to at least one other node in the network.
For the remaining nodes, the transition is more gradual, reflecting larger differences in degree and triangle count.
This transition is fastest for the ER and WS models in which node degrees are less diverse.
For the BA model, in which some nodes have a substantially larger degree or triangle count than others as a result of the power-law degree distribution, the transition is slower and more gradual.
In denser graphs, as indicated by $m$, the value of $\phi$ required to make most nodes anonymous initially increases, because the structure surrounding the nodes becomes more distinct.
At even higher densities, the required $\phi$ value decreases.
While nodes remain structurally distinctive, the higher node degrees and triangle counts allow greater deviations in node and triangle count.

The case $k=8$, shown in the bottom subfigures of Fig.~\ref{fig:fuzzy-models}, represents a stronger privacy requirement where a node is considered anonymous when it is $\phi$-similar to at least seven other nodes.
The results show that low values of $\phi$ greatly increase anonymity in this setting as well.
In fact, $\phi=10\%$ makes the majority of nodes $8$-anonymous across most graph models and density settings.
Compared to $k=2$, higher levels of uncertainty $\phi$ are required to make most nodes anonymous, and the transition happens more gradually.
The required value of $\phi$ peaks around $m=10$, after which it decreases as larger differences in degree and triangle count are allowed in denser networks.
Hence, even relatively low levels of uncertainty suffice to achieve stronger privacy guarantees, indicated by a higher value of $k$. 
Most importantly, our results on synthetic graphs suggest that even small levels of uncertainty can substantially reduce node uniqueness, even in dense networks.

\subsection*{Fuzzy $k$-anonymity in real-world networks}\label{sub:fuzzyreal}
To assess $\phi$-$k$-anonymity in real-world networks, we focus on four uncertainty levels: no uncertainty ($\phi=0\%$), very low uncertainty ($\phi=1\%$), low uncertainty ($\phi=5\%$), and a higher level of uncertainty ($\phi=10\%$).
The subfigure in the top row of Fig.~\ref{fig:anon:realnws} shows the fraction of $\phi$-$k$-anonymous nodes in the real-world network datasets for each value of $\phi$, indicated by color.
As some networks contain a high fraction of initially anonymous nodes in the non-fuzzy setting ($\phi=0\%$), the bottom subfigure in Fig.~\ref{fig:anon:realnws} shows the fraction of nodes that are unique in the non-fuzzy setting that become anonymous when accounting for the indicated level of uncertainty $\phi$.
A value of 1.0 indicates that all unique nodes become anonymous, and a value of 0.0 indicates no increase in anonymity.

\begin{figure}[t!]
    \centering
    \includegraphics[width=\textwidth]{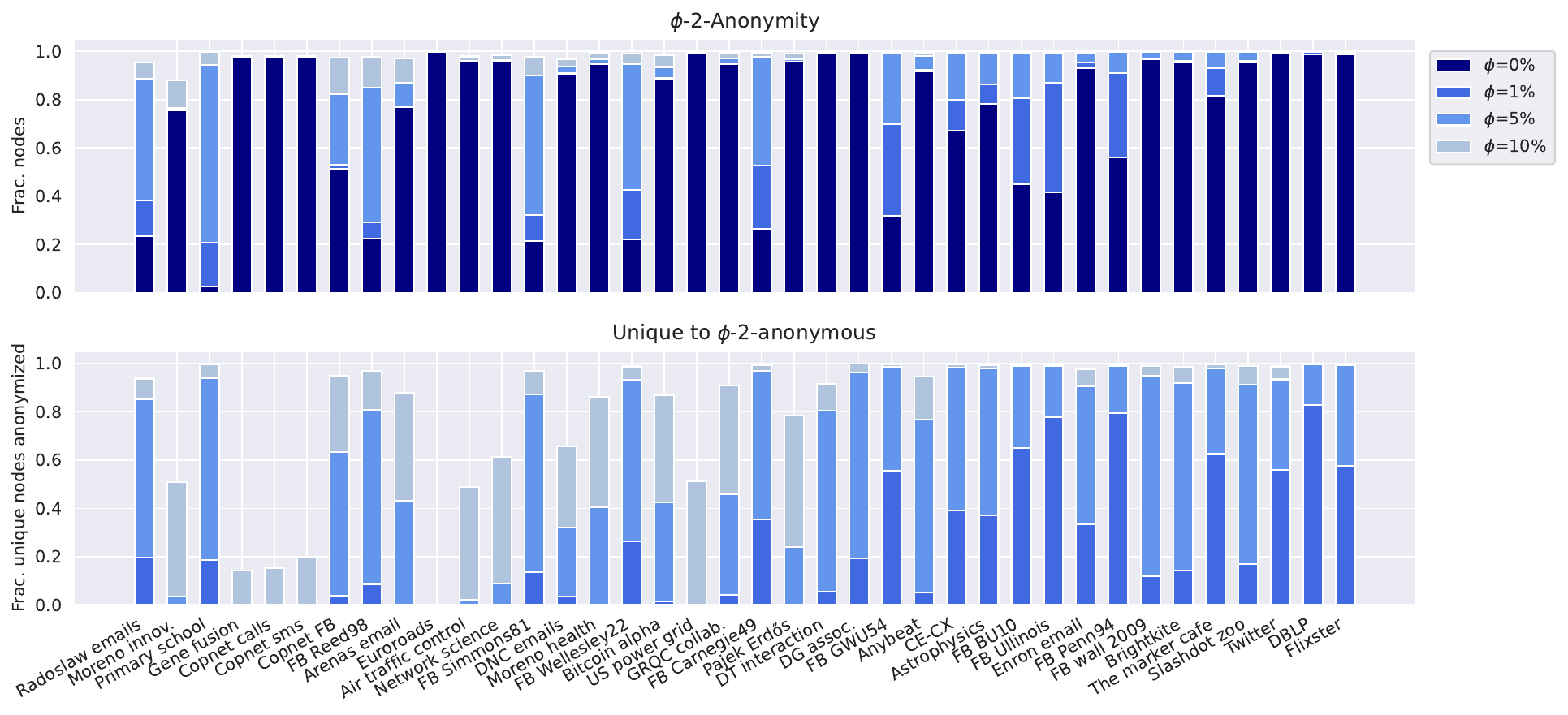}
    \caption{$\phi$-$2$-Anonymity in real-world networks (sorted by their number of nodes). Top: the fraction of $\phi$-$2$-anonymous nodes (vertical axis) in the networks (horizontal axis) for different values of $\phi$, as indicated by color.
    Bottom: Fraction of unique nodes that becomes anonymous when accounting for $\phi\%$ uncertainty (indicated by color) compared to $\phi=0\%$.}
    \label{fig:anon:realnws}
\end{figure}

Overall, we find that $\phi=5\%$ results in 2.5 times more anonymous nodes compared to $\phi=0\%$, rendering 64\% of previously unique nodes anonymous, with at least 76\% of nodes anonymous in all datasets.
With a larger level of uncertainty, $\phi=10\%$, more than 88\% of the nodes are anonymous and, on average, 82\% of the initially unique nodes become anonymous.
While a low level of uncertainty, $\phi=1\%$, has an overall small effect on the anonymity of networks, this still renders on average 22\% of the initially unique nodes anonymous. 
Especially in large networks, the effect is substantial.

The bottom subfigure shows that $\phi=5\%$ substantially increases anonymity for small networks with low initial anonymity, such as ``Radoslaw emails'', ``Primary school'' and ``FB Reed98''. 
For other small networks with many low degree nodes (degree and triangle count distributions can be found in Supplementary Information) anonymity increases only for $\phi=10\%$, indicating that these networks require a higher level of uncertainty to obtain more anonymity.
In contrast, the seventeen largest networks (rightmost networks in Fig.~\ref{fig:anon:realnws}) exhibit a different behavior. 
Although these large networks have high initial anonymity, $\phi=5\%$ is often sufficient to render most of the remaining unique nodes anonymous, achieving results close to full anonymization.

Overall, the findings suggest that incorporating low levels of uncertainty in attacker knowledge can substantially increase node anonymity across a variety of real-world networks.
Building on these insights, we next investigate how the application of anonymization algorithms can further enhance network anonymity.

\subsection*{Fuzzy anonymization in graph models}\label{sub:anonmodel}
In this section, we assess how introducing uncertainty in the measurement of anonymity affects the performance of anonymization algorithms in synthetic graph models.
Given our focus on the budgeted variant of the anonymization problem,~\cite{de2026anonymization} we delete edges with the \textsc{es}, \textsc{ua} and \textsc{greedy} anonymization algorithms until the graph is $\phi$-2-anonymous, or until 5\% of the edges are deleted.
For completeness, Supplementary Information contains results showing how these findings generalize to larger graph models with 1,000 nodes, for which similar trends hold.

Figure~\ref{fig:alg:models} shows anonymity in networks generated using the ER, BA and WS graph models after applying anonymization algorithms.
The colors indicate the anonymization algorithm used, whereas each linestyle indicates the level of uncertainty $\phi$ accounted for.
Across all graph models, even a small level of uncertainty ($\phi=5\%$) substantially reduces the fraction of unique nodes after anonymization.
This improvement is partly because, as shown in Fig.~\ref{fig:fuzzy-models}, accounting for uncertainty already increases anonymity substantially before anonymization.
Although for $\phi=0\%$ we observe that uniqueness increases considerably with network density, this effect is largely mitigated when introducing a low level of uncertainty ($\phi=5\%$).
While $\phi=1\%$ results in uniqueness levels similar to the non-fuzzy setting, $\phi=5\%$ combined with \textsc{greedy} anonymization leads to graphs that are close to fully anonymous across all models and density settings.
In general, introducing small levels of uncertainty makes it feasible to even anonymize dense graphs in which most nodes are unique when $\phi=0\%$,~\cite{romanini2021privacy, dejong2023effect} and standard anonymization algorithms often fail.

\begin{figure}[t!]
    \centering
    \includegraphics[width=\textwidth]{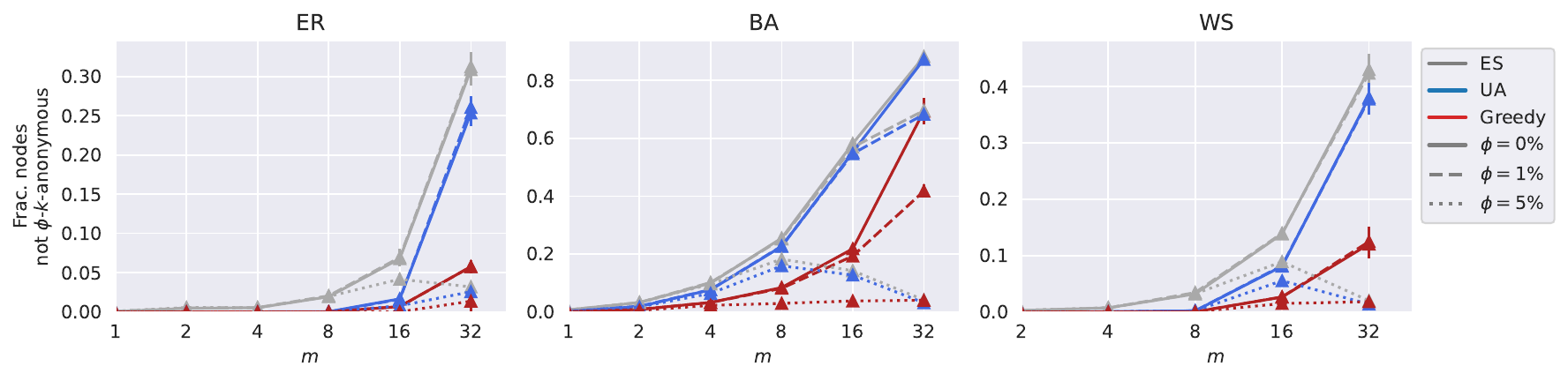}
    \caption{Fuzzy $\phi$-$2$-Anonymity after budgeted anonymization (deleting 5\% of the edges) on graph models ER (left), BA (middle) and WS (right) with 500 nodes and varying density (horizontal axis) using different anonymization algorithms (color) and levels of uncertainty $\phi$ (linestyle). }
    \label{fig:alg:models}
\end{figure}

\begin{figure}[b!]
    \centering
    \includegraphics[width=\textwidth]{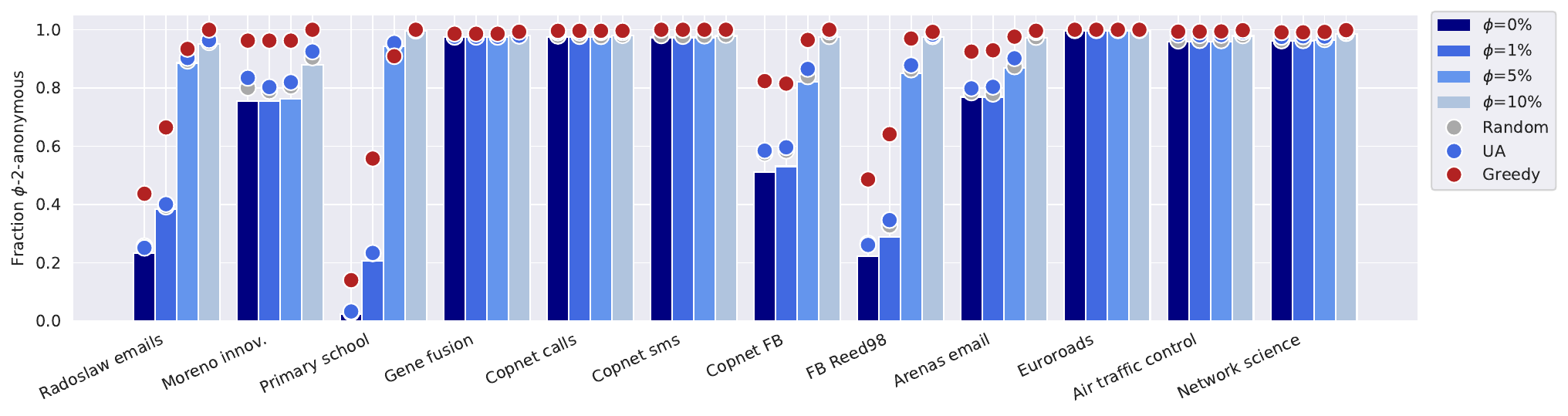}
    \caption{Budgeted anonymization in real-world networks with different levels of uncertainty $\phi$. Each bar indicates the fraction of $\phi$-$2$-anonymous nodes (vertical axis) for the network and corresponding $\phi$ value (color). Each dot denotes the fraction of $\phi$-$2$-anonymous nodes after applying the anonymization algorithm indicated by that color.}
    \label{fig:anonnws}
\end{figure}

\subsection*{Fuzzy anonymization in real-world networks}\label{sub:anonreal}
To assess the effect of anonymization when accounting for uncertainty in real-world networks, Fig.~\ref{fig:anonnws} shows the fraction of $\phi$-$2$-anonymous nodes after applying anonymization algorithms for 12 real-world networks.
Similar to the previous experiment, the anonymization algorithms, \textsc{es}, \textsc{ua} and \textsc{greedy}, delete at most 5\% of the edges for each network.
As anonymity before anonymization differs across the included networks, the fraction of $\phi$-$2$-anonymous nodes before anonymization is shown by the solid bars in Fig.~\ref{fig:anonnws}, while the dots represent the anonymity after anonymization.

Across all settings, the \textsc{greedy} anonymization algorithm is the most effective, yet insufficient to anonymize all nodes in some of the included networks.
For $\phi=0\%$, \textsc{greedy} increases the number of $\phi$-$2$-anonymous nodes by 12\% on average and by up to 30\% for ``Copnet FB''.
Combined with $\phi=5\%$, \textsc{greedy} anonymization makes over 90\% of nodes $\phi$-$2$-anonymous for all included networks.
With $\phi=10\%$, it manages to anonymize close to all nodes (at least 99\%) for the included networks.
These results show that, while the considered algorithms yield limited gains within a 5\% budget, the combination of \textsc{greedy} anonymization with a modest level of uncertainty ($\phi=10\%$) is sufficient to nearly fully anonymize many real-world networks.

\begin{figure}[t!]
    \centering
    \includegraphics[width=0.9\textwidth]{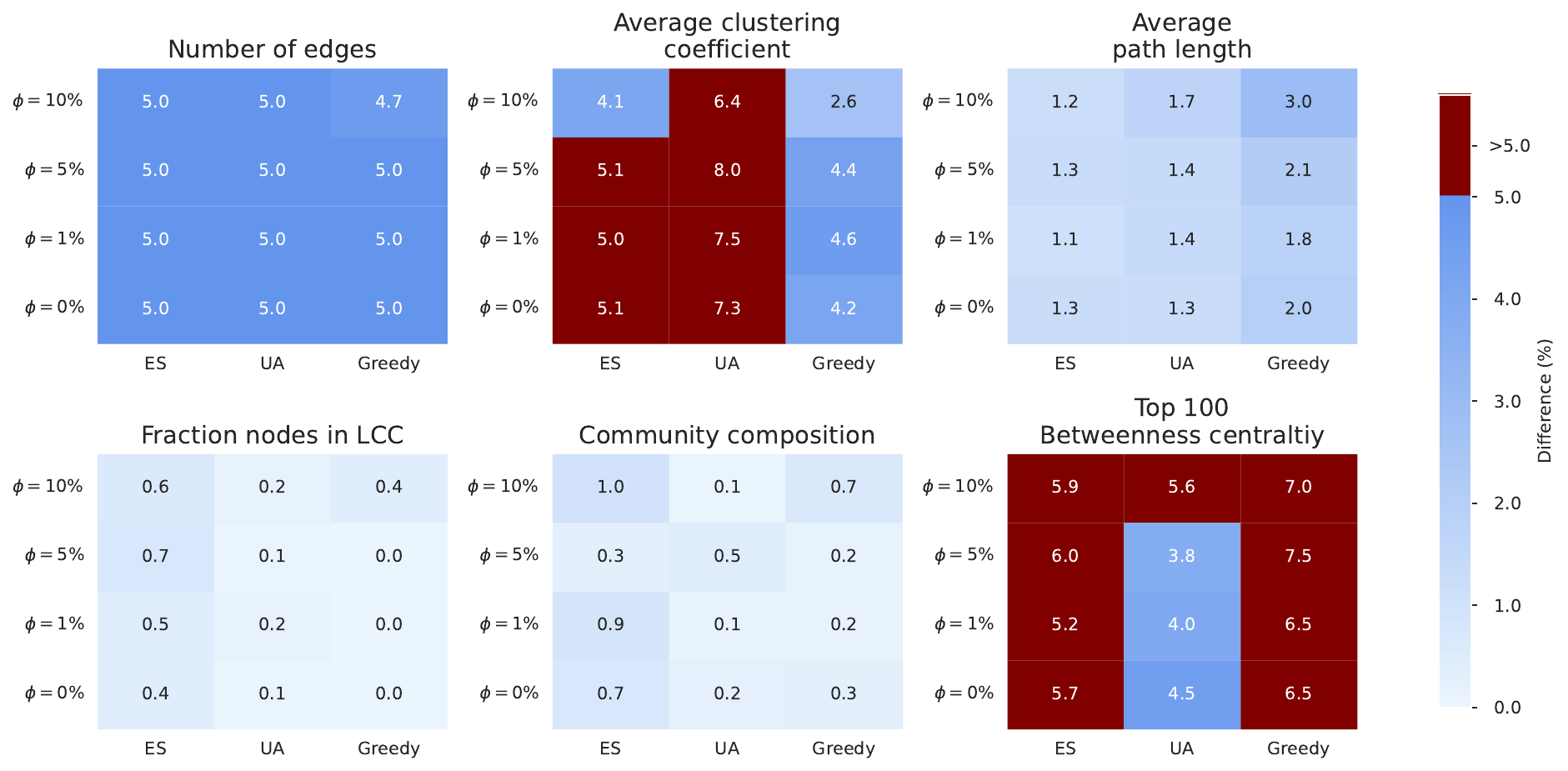}
    \caption{Data utility after budgeted anonymization with the \textsc{es}, \textsc{ua} and \textsc{greedy} anonymization algorithms. Each cell shows the relative difference, compared to the network before anonymization, for one of the six network utility metrics. Results are shown for the three anonymization algorithms (horizontal axis), and four levels of uncertainty $\phi$ (vertical axis). White indicates no difference, while the darker shades of blue indicate larger differences. Red cells indicate differences larger than 5\%.}
    \label{fig:utility}
\end{figure}

\subsection*{Utility of anonymized real-world networks}\label{sub:utility}
While the previous sections show that fuzzy anonymization can substantially improve anonymity within the given budget of edge deletions, it is also important to evaluate how well these algorithms preserve data utility when different levels of uncertainty $\phi$ are considered.
Figure~\ref{fig:utility} shows, for the 12 networks included in the anonymization experiments, how well utility is preserved.
The subfigures in the top row and the bottom left subfigure show the difference in four structural network properties: (1) the number of edges, (2) average clustering coefficient, (3) average path length, and (4) fraction of nodes remaining in the largest connected component.
The two rightmost subfigures in the bottom row show the performance on common network analysis tasks: (5) how well community structure is preserved, and (6) the similarity of the top 100 most central nodes according to betweenness centrality.
If the percentage difference in a metric exceeds 5\%, this is indicated in red to show that an undesirable amount of data utility is lost.
As some networks exhibit outliers, especially in centrality, we report the median across networks.
The~\hyperref[sec:methods]{Methods} section describes how each metric is computed; results for each network separately are included in Supplementary Information.

The top left subfigure, showing the fraction of edges deleted, indicates that all algorithms in all settings, except for \textsc{greedy} with $\phi=10\%$, use the full budget of 5\%, as they do not manage to fully anonymize the graphs within the budget.
Hence, for all $\phi$ values, the fraction of edges deleted often remains the same, while the set of nodes that requires anonymization changes.

Overall, the average path length, largest component size and community structure are well preserved by all anonymization algorithms at all levels of uncertainty.
Clustering and centrality prove to be more challenging to preserve.
For clustering, \textsc{greedy} performs best and \textsc{ua} performs worst, though higher values of $\phi$ generally improve preservation. 
Centrality is particularly difficult to preserve: even random edge deletions by \textsc{es} lead to substantial changes in the top 100 most central nodes, and at $\phi=10\%$ all algorithms exceed the 5\% difference threshold.

Overall, higher levels of uncertainty do not substantially reduce data utility for most metrics when using the \textsc{greedy} algorithm, except for centrality, suggesting that $\phi$-$2$-anonymity still reasonably balances privacy and utility.  

\section*{Discussion and conclusion}\label{sec:discussion}
In this paper, we argued that attackers attempting to extract sensitive information from a published network dataset are unlikely to possess exact knowledge of the network's structural information.
To account for uncertainty in the assumed attacker scenario, we introduced $\phi$-$k$-anonymity, a fuzzy variant of $k$-anonymity.
While the proposed framework can be applied to various $k$-anonymity measures, we focus on the $(deg, tri)$ measure due to its expressiveness, relative ease of computation, and interpretability in the fuzzy setting.

Our experiments on both graph models and real-world networks demonstrated that even a modest level of uncertainty substantially increases the number of anonymous nodes.
This finding suggests that many structurally unique nodes are highly similar to at least one other node in the network, which an attacker would be unable to distinguish if their information is not exact.
Across graph models with varying densities, a small level of uncertainty ($\phi=5\%$) led to a transition from almost all nodes being unique to nearly all nodes being anonymous. 
Moreover, in real-world networks, introducing $\phi=5\%$ made 64\% of the initially unique nodes anonymous.
In particular, both small networks with low anonymity in the non-fuzzy setting and large networks in our benchmark datasets benefit substantially from such uncertainty.

When applying anonymization algorithms with a perturbation budget of 5\% of the edges, accounting for uncertainty further enhances anonymity, often rendering most nodes anonymous.
In dense graph models, where uniqueness remains high even after anonymization, incorporating uncertainty mitigates this effect, allowing most nodes to be anonymized.
While anonymization algorithms have only limited impact on most real-world networks in the non-fuzzy setting, \textsc{greedy} anonymization with $\phi=10\%$ anonymizes nearly all nodes across the 12 evaluated networks in the considered setting.
Additionally, data utility remained largely unaffected even at higher levels of uncertainty, retaining similar values for structural properties and performance on network analysis tasks.
Altogether, allowing for uncertainty in the attacker scenario makes the publication of $k$-anonymous networks much more attainable, even for networks that are very difficult to anonymize in the non-fuzzy setting.

Our findings illustrate that assumptions about attacker knowledge play a critical role in determining achievable privacy guarantees for network data and the resulting trade-off between privacy and data utility.
Future research should investigate suitable values of $\phi$ for different types of networks and investigate how the application of anonymization algorithms or the addition of noise to the network can increase the assumed level of uncertainty.
Additionally, since information beyond a node's direct neighborhood can be highly de-anonymizing,~\cite{dejong2023effect} it is important to study fuzzy variants of measures that account for information reaching beyond the direct neighborhood of a node.
Such research should consider potential cascading effects that may occur if an attacker uniquely identifies a node.

Overall, our results show that with modest levels of uncertainty, both anonymity and high utility are increasingly attainable, motivating further research on uncertainty-aware privacy guarantees for network data.

\section*{Methods}\label{sec:methods}
In this section we provide definitions, notation, experimental setup, descriptions of anonymization algorithms, and utility metrics used.

\subsection*{Definitions and notation}
We define an undirected graph $G=(V, E)$ consisting of a node set $V$ and an edge set $E \subseteq \{\{v, w\} : v, w \in V\}$.
The degree of a node $v$ is the number of neighbors: $deg(v) = |\{w : \{v, w\} \in E\}|$. 
The distance $dist(v, w)$ between two nodes $v, w \in V$ is the length of the shortest path between them, i.e., a path with a minimum number of edges traversed.
By convention $dist(v, v) = 0$, and  $dist(v, w)=\infty$ if no path exists between $v$ and $w$.
The latter refers to the case in which the nodes are in different connected components; maximal sets of nodes in which each pair of nodes is connected through a path.
Most nodes in the network are typically in the largest connected component (LCC).
The average path length equals the average of $dist(v, w)$ over all pairs $v \neq w \in V$ with $dist(v, w) < \infty$. 
The $d$-neighborhood of a node $v$, denoted $N_d(v) = (V_{N_d(v)}, E_{N_d(v)})$, is the subgraph induced by all nodes within distance $d$ from $v$ where $V_{N_d(v)} = \{w : dist(v, w) \leq d\}$ and $E_{N_d(v)} = \{ \{u, w\} : u, w \in V_{N_d(v)} \}$.

The number of triangles incident to $v$ equals $tri(v)=|\{ \{u, w\} \in E : u, w \in V_{N_1(v)} \setminus \{v\}\}|$.
The clustering coefficient of a node is $c(v) = tri(v) / \binom{deg(v)}{2}$, the fraction of pairs of neighbors of $v$ that are connected.
The average clustering coefficient of a graph $G$ is $C(G) = \frac{1}{|V|}\sum_{v\in V}c(v)$.

Nodes in real-world networks often cluster into communities: groups that are more densely connected internally than to the rest of the network.
We find communities by applying the Leiden community detection algorithm.~\cite{traag2019louvain}
The quality of a partition into communities is measured by modularity.
Node importance can be quantified using centrality measures.
The betweenness centrality of a node $v$ is the fraction of shortest paths between all pairs of nodes that pass through $v$.

To measure node anonymity, we use the notion of $k$-anonymity.
Two nodes $v, w \in V$ are equivalent if they have the same $(deg, tri)$ signature.
A node is considered $k$-anonymous, if there are at least $k-1$ nodes in the network with the same signature.
Similarly, a node $v$ is $\phi$-$k$-anonymous, c.f. Equation~\ref{eq:phisim}, if it is $\phi$-similar to at least $k-1$ other nodes.
We assess the anonymity of the network as a whole as the fraction of nodes that is $\phi$-$k$-anonymous: $\frac{|\{v\in V : |\phi\text{-}similar(v)| \geq 2\}|}{|V|}$.
A node is unique if it is not $\phi$-$2$-anonymous. 

\subsection*{Anonymization algorithms}
We use three anonymization algorithms: 
(1) Edge Sampling (\textsc{es})~\cite{romanini2021privacy} 
(2) Unique Affected (\textsc{ua})~\cite{de2026anonymization}
and (3) a newly proposed \textsc{greedy} anonymization algorithm which works as follows.
During each iteration, \textsc{greedy} computes the marginal anonymity improvement obtained by deleting every edge. 
The edges are then sorted in descending order of this improvement, and the top $R$ edges are selected for deletion.
If the algorithm finds an edge which anonymizes all nodes, only that edge is deleted. 
If multiple edges yield the same anonymity improvement, ties are broken by selecting the edge incident to the largest number of unique (non-anonymous) nodes (0, 1 or 2).
In case ties still remain, edges are ordered lexicographically.

To determine the anonymity gain obtained by removing a particular edge, the algorithm computes the resulting fraction of $\phi$-$k$-anonymous nodes for each possible edge deletion.
To perform these computations efficiently, the algorithm maintains the signature of each node and whether it is currently unique or not.
For each edge, the algorithm updates the signatures of nodes for which the signature changes  upon deleting the edge.
Based on the new signatures, it recomputes the similarity between all node pairs in the network, based on which the fraction of $\phi$-$k$-anonymous nodes is derived.
The algorithm continues until either the budget is depleted or all nodes are anonymized.

\subsection*{Utility metrics}\label{method:utility}
To assess data utility of the resulting anonymized network, we take into account six metrics. The first four metrics relate to how well the structure of the network is preserved:
\begin{itemize}
    \item The fraction of edges deleted.
    \item The change in the average clustering coefficient.
    \item The change in average path length.
    \item The change in the fraction of nodes that is part of the largest connected component (LCC).
\end{itemize}
For each of these measures, we first compute the value for the original network $G$ and the anonymized network $G'$, and report the relative change compared to the original network.
The last two metrics concern network analysis tasks for which the network could be used: 
\begin{itemize}
    \item The change in community structure. 
    \item The fraction of nodes that are no longer in the top 100 most central nodes according to betweenness centrality.
\end{itemize}

To determine how well community structure is preserved, we use the Normalized Mutual Information (NMI) score~\cite{lancichinetti2012consensus} comparing the partition of nodes into communities found by the \textsc{leiden} algorithm~\cite{traag2019louvain} before and after anonymization.
The communities obtained by community detection algorithms are often not stable: when running community detection multiple times on the same network, the NMI between those communities found does not necessarily equal 1.0. 
We account for this phenomenon as follows.
For the original network and the anonymized network we compute the communities $20$ times using the \textsc{leiden} algorithm.
Based on this we determine $NMI_{stability}$, comparing the communities found for the original network, and $NMI_{anon}$, comparing the communities found for the anonymized network to those found for the original network. 
These metrics are defined in equations~(\ref{eq:nmistability}) and~(\ref{eq:nmianon}) where $C = \{C_1, ..., C_{|C|}\}$ denotes the set of community assignments for the original network and $C'$ the community assignments for the anonymized network. $|C|$ denotes the number of community assignments.

\begin{equation}
    NMI_{stability} = \frac{2}{|C| \cdot (|C| -1)} \sum_{1 \leq i < j \leq |C|} NMI(C_i, C_j)
    \label{eq:nmistability}
\end{equation}

\begin{equation}
    NMI_{anon} = \frac{1}{|C| \cdot |C'|} \sum^{|C|}_{i=1} \sum^{|C'|}_{j=1} NMI(C_i, C'_j) 
    \label{eq:nmianon}
\end{equation}

If $NMI_{stability}$ is low, it is expected that $NMI_{anon}$ will also be low. 
To account for this, we report the difference between the two as defined in equation~(\ref{eq:nmiutil}). 
Finally, note that due to the non-determinism of the community detection algorithms, it could occur that $NMI_{anon}$ is larger than $NMI_{stability}$, which would indicate that the communities generated for the anonymized network ($C'$) have, on average, more in common with the communities found for the original network ($C$), than the communities in $C$ have with themselves. This results in negative values $NMI_{stability} - NMI_{anon}$. 
Experiments in Supplementary Information show that these negative values occur for real-world networks after anonymization. 
This is likely due to chance as the differences between $NMI_{stability} - NMI_{anon}$ are usually close to zero, and these negative values occur more regularly for networks with a low $NMI_{stability}$.

\begin{equation}
    NMI_{utility} = max(0.0, NMI_{stability} - NMI_{anon}) 
    \label{eq:nmiutil}
\end{equation}

\subsection*{Experimental setup and data}
For our experiments, we use three commonly used synthetic graph models which each mimic different real-world properties, and a wide range of real-world networks contained in Table~\ref{tab:data}.
For graph models, we use the Erdős Rényi model (ER),~\cite{erdos1960evolution} the Barabási Albert model (BA),~\cite{barabasi1999emergence} and Watts Strogatz model (WS)~\cite{watts1998collective} with rewiring probability $0.05$.
The ER and BA graphs are generated using igraph,~\cite{igraph} the WS graphs with NetworkX.~\cite{networkx}
For the network models we use $|V|=500$ and a number of edges per node $m=\{1, 2, \dots, 10, 15, \dots, 55, 60\}$ for experiments concerning anonymity, and $m=\{1, 2, 4, 8, 16, 32\}$ for experiments concerning anonymization.

To account for non-determinism in the generation of the networks, we average results on anonymity in graph models over 10 generated networks for each combination of model and number of edges.
Results concerning anonymization are averaged over 5 generated graphs for each model, on which the anonymization algorithm with each setting for $\phi$ is run 5 times.
Each anonymization algorithm deletes 5\% of all edges. Hence budget $B=\lfloor 0.05 \cdot |E|\rfloor $.
To account for possibly long runtimes of the \textsc{greedy} algorithm, we use a recompute gap equal to 1/20th of the edges to be deleted. Hence, each iteration of the algorithm deletes $R=\lfloor 1/20\cdot B \rfloor $ edges.
For experiments on real-world networks, we report results for the 12 networks for which the \textsc{greedy} algorithm terminates within a time limit of 3 hours.

All used code is available via GitHub.
Code for measuring anonymity and perform anonymization is implemented in C++, with separate repositories for \textsc{es} and \textsc{ua} anonymization (\href{https://github.com/RacheldeJong/ANONET}{github.com/RacheldeJong/ANONET}) and  \textsc{greedy} anonymization (\href{https://github.com/franktakes/optianon}{github.com/franktakes/optianon}).

\section*{Acknowledgements}
We would like to thank the Network Science group     (\hyperlink{www.networkscience.nl}{www.networkscience.nl}), especially Gamal Adel Elgamal, for various helpful suggestions and discussions. 

\section*{Author contributions}
R.J., M.L., and F.T. conceptualized the study and methodology. R.J. wrote the manuscript. 
M.L. and F.T.  
supervised and administered the project. R.J. developed the software, performed
the experiments and visualized the results. F.T. developed the greedy algorithm. All authors reviewed and approved the manuscript.

\section*{Data availability}
All network datasets are available in the repositories or accompanying the paper cited in Table~\ref{tab:data}.

\bibliography{bibliography}

@InProceedings{arsene2026simulated,
    author="Arsene, E. Denisa
    and de Jong, Rachel G.
    and Takes, Frank W.
    and Latour, Anna L. D.",
    title="A Simulated Annealing Approach to Social Network Anonymization",
    booktitle="Proceedings of Complex Networks {\&} Their Applications XIV",
    year="2026",
    pages="205--216",
}

@Article{azizi2020epidemics,
  title={Epidemics on networks: Reducing disease transmission using health emergency declarations and peer communication},
  author={Azizi, Asma and Montalvo, Cesar and Espinoza, Baltazar and Kang, Yun and Castillo-Chavez, Carlos},
  journal={Infectious Disease Modelling},
  volume={5},
  pages={12--22},
  year={2020},
  p2ublisher={Elsevier}
}

@Article{barabasi1999emergence,
  title={Emergence of scaling in random networks},
  author={Barab{\'a}si, Albert L{\'a}szl{\'o} and Albert, R{\'e}ka},
  journal={Science},
  volume={286},
  number={5439},
  pages={509--512},
  year={1999},
  p2ublisher={American Association for the Advancement of Science},
}

@misc{biosnapnets,
  author={Marinka Zitnik and Rok Sosi\v{c} and Sagar Maheshwari and Jure Leskovec},
  title={BioSNAP Datasets: Stanford. Biomedical Network Dataset Collection},
  howpublished ={\url{http://snap.stanford.edu/biodata} (last accessed 2025)},
  year = {2018},
}

@article{blumenthal2020exact,
  title={On the exact computation of the graph edit distance},
  author={Blumenthal, David B and Gamper, Johann},
  journal={Pattern Recognition Letters},
  volume={134},
  pages={46--57},
  year={2020},
  publisher={Elsevier}
}

@Article{bojanowski2014measuring,
  title={Measuring segregation in social networks},
  author={Bojanowski, Micha{\l} and Corten, Rense},
  journal={Social Networks},
  volume={39},
  pages={14--32},
  year={2014},
  p2ublisher={Elsevier}
}

@article{bokanyi2023anatomy,
  title={The anatomy of a population-scale social network},
  author={Bok{\'a}nyi, Eszter and Heemskerk, Eelke M and Takes, Frank W},
  journal={Scientific Reports},
  volume={13},
  number={1},
  pages={9209},
  year={2023},
  p2ublisher={Nature Publishing Group UK London}
}

@inproceedings{bonello2025ga,
  author       = {Samuel Bonello and
                  Rachel G. de Jong and
                  Thomas H. W. B{\"{a}}ck and
                  Frank W. Takes},
  title        = {Utility-aware Social Network Anonymization using Genetic Algorithms},
    year = {2025},
    isbn = {9798400714641},
    pages = {775–778},
    numpages = {4},
    booktitle = {Proceedings of the Genetic and Evolutionary Computation Conference Companion (GECCO)},
}

@article{cremers2025unveiling,
  title={Unveiling the social fabric through a temporal, nation-scale social network and its characteristics},
  author={Cremers, Jolien and Kohler, Benjamin and Maier, Benjamin Frank and Eriksen, Stine Nymann and Einsiedler, Johanna and Christensen, Frederik K{\o}lby and Lehmann, Sune and Lassen, David Dreyer and Mortensen, Laust Hvas and Bjerre-Nielsen, Andreas},
  journal={Scientific Reports},
  volume={15},
  number={1},
  pages={18383},
  year={2025},
  publisher={Nature Publishing Group UK London}
}

@misc{dataforgoodlab,
 title={Data 4 good lab},
 author={Michael Fire},
 year ={2020},
 howpublished ={\url{https://data4goodlab.github.io/MichaelFire/\#section3} (last accessed 2025)}
}

@inproceedings{de2026anonymization,
  title={The anonymization problem in social networks},
  author={de Jong, Rachel G. and van der Loo, Mark P. J. and Takes, Frank W.},
  booktitle={Proceedings of the Workshop on Modelling and Mining Networks (WAW)},
  year={To appear}
}

@techreport{deVries2023risk,
  author       = {Marieke M. de Vries and Rachel G. de Jong and Mark P. J. van der Loo and Peter-Paul de Wolf and Frank W. Takes},
  title        = {The Risk of Identity Disclosure through Network Structure: Anecdotal Evidence from a Hackathon},
  institution  = {United Nations Economic Commission for Europe},
  type         = {Working Document},
  year         = {2023},
  month        = aug,
  url          = {https://unece.org/statistics/documents/2023/08/working-documents/risk-identity-disclosure-through-network-structure},
}

@article{dejong2023algorithms,
    author = {de Jong, Rachel G. and van der Loo, Mark P. J. and Takes, Frank W.},
    title = {Algorithms for Efficiently Computing Structural Anonymity in Complex Networks},
    year = {2023},
    issue_date = {2023},
    p2ublisher = {Association for Computing Machinery},
    address = {New York, NY, USA},
    volume = {28},
    issn = {1084-6654},
    journal = {ACM Journal of Experimental Algorithmics},
    articleno = {1.7},
    numpages = {22}
}

@article{dejong2023effect,
  title={The effect of distant connections on node anonymity in complex networks},
  author={de Jong, Rachel G. and van der Loo, Mark P. J. and Takes, Frank W.},
  journal={Scientific Reports},
  volume={14},
  number={1},
  pages={1156},
  year={2024},
  p2ublisher={Nature Publishing Group UK London}
}

@article{dejong2024comparison,
  title={A systematic comparison of measures for k-anonymity in networks},
  author={de Jong, Rachel G and van der Loo, Mark PJ and Takes, Frank W},
  journal={arXiv preprint arXiv:2407.02290},
  year={2024}
}

@article{duncan_disclosure,
	title = {Disclosure-{Limited} {Data} {Dissemination}},
	volume = {81},
	issn = {0162-1459},
	url = {https://www.tandfonline.com/doi/abs/10.1080/01621459.1986.10478229},
	number = {393},
	urldate = {2023-09-18},
	journal = {Journal of the American Statistical Association},
	author = {Duncan, George T. and Lambert, Diane},
	month = mar,
	year = {1986},
	pages = {10--18},
}

@Article{erdos1960evolution,
  title={On the evolution of random graphs},
  author={Erdős, Paul and R{\'e}nyi, Alfr{\'e}d},
  journal={Publication of the Mathematical Institute of the Hungarian Academy of Sciences},
  volume={5},
  number={1},
  pages={17--60},
  year={1960},
  p2ublisher={Citeseer}
}

@Inproceedings{hay2008resisting,
    author = {Hay, Michael and Miklau, Gerome and Jensen, David and Towsley, Don and Weis, Philipp},
    title = {Resisting structural re-identification in anonymized social networks},
    year = {2008},
    issue_date = {August 2008},
    p2ublisher = {VLDB Endowment},
    volume = {1},
    number = {1},
    issn = {2150-8097},
    booktitle = {Proceedings of the VLDB Endowment},
    month = {aug},
    pages = {102–114},
    numpages = {13}
}

@Article{igraph,
  title = {The igraph software package for complex network research},
  author = {Gabor Csardi and Tamas Nepusz},
  journal = {InterJournal Complex Systems},
  volume={1695},
 
  year = {2006},
  url = {https://igraph.org},
}

@article{jiang2021applications,
  title={Applications of differential privacy in social network analysis: A survey},
  author={Jiang, Honglu and Pei, Jian and Yu, Dongxiao and Yu, Jiguo and Gong, Bei and Cheng, Xiuzhen},
  journal={IEEE Transactions on Knowledge and Data Engineering},
  volume={35},
  number={1},
  pages={108--127},
  year={2021},
  p2ublisher={IEEE}
}

@article{kazmina2024socio,
  title={Socio-economic segregation in a population-scale social network},
  author={Kazmina, Yuliia and Heemskerk, Eelke M and Bok{\'a}nyi, Eszter and Takes, Frank W},
  journal={Social Networks},
  volume={78},
  pages={279--291},
  year={2024},
  p2ublisher={Elsevier}
}

@inproceedings{kunegis2013konect,
    author = {Kunegis, J\'{e}r\^{o}me},
    title = {KONECT: the {K}oblenz network collection},
    year = {2013},
    isbn = {9781450320382},
    booktitle = {Proceedings of the 22nd International Conference on World Wide Web},
    pages = {1343–1350},
    numpages = {8},

}

@article{lambert1993measures,
  title={Measures of disclosure risk and harm},
  author={Lambert, Diane},
  journal={Journal of Official Statistics},
  volume={9},
  pages={313--331},
  year={1993},
}

@article{lancichinetti2012consensus,
  title={Consensus clustering in complex networks},
  author={Lancichinetti, Andrea and Fortunato, Santo},
  journal={Scientific Reports},
  volume={2},
  number={1},
  pages={336},
  year={2012},
  publisher={Nature Publishing Group UK London}
}

@article{li2018influence,
  title={Influence maximization on social graphs: A survey},
  author={Li, Yuchen and Fan, Ju and Wang, Yanhao and Tan, Kian-Lee},
  journal={IEEE Transactions on Knowledge and Data Engineering},
  volume={30},
  number={10},
  pages={1852--1872},
  year={2018},
  p2ublisher={IEEE}
}

@Inproceedings{liu2008towards,
    author = {Liu, Kun and Terzi, Evimaria},
    title = {Towards identity anonymization on graphs},
    year = {2008},
    isbn = {9781605581026},
    booktitle = {Proceedings of the ACM SIGMOD International Conference on Management of Data},
    pages = {93–106},
    numpages = {14},

}

@inproceedings{networksrepository,
    author = {Rossi, Ryan A. and Ahmed, Nesreen K.},
    title = {The network data repository with interactive graph analytics and visualization},
    year = {2015},
    isbn = {0262511290},
    
    booktitle = {Proceedings of the AAAI Conference on Artificial Intelligence},
    pages = {4292–4293},
    numpages = {2},
}

@techreport{networkx,
  title={Exploring network structure, dynamics, and function using NetworkX},
  author={Aric Hagberg and Pieter Swart and Daniel Schult},
  year={2008},
  institution={Los Alamos National Lab (LANL)}
}

@article{panayiotou2025anatomy,
  title={Anatomy of a Swedish population-scale network: G. Panayiotou et al.},
  author={Panayiotou, Georgios and Wohlert, Inga K and Bask, Miia and Bask, Mikael and Magnani, Matteo and M{\"a}kinen, Ilkka Henrik},
  journal={Scientific Reports},
  volume={15},
  number={1},
  pages={30300},
  year={2025},
  publisher={Nature Publishing Group UK London}
}

@article{romanini2021privacy,
  title={Privacy and uniqueness of neighborhoods in social networks},
  author={Romanini, Daniele and Lehmann, Sune and Kivel{\"a}, Mikko},
  journal={Scientific Reports},
  volume={11},
  number={1},
  pages={20104},
  year={2021},
  p2ublisher={Nature Publishing Group UK London}
}

@misc{sapiezynski2019copenhagen,
  title={The Copenhagen Networks Study interaction data. Figshare},
  author={Piotr Sapiezynski and Arkadiusz Stopczynski and David D. Lassen and Sune L. Jørgensen },
  howpublished={\url{https://doi.org/10.6084/m9.figshare.7267433.v1} (last accessed 2025) },
  year={2019},
}

@misc{snapnets,
  author= {Jure Leskovec and Andrej Krevl},
  title= {SNAP Datasets: Stanford Large Network Dataset Collection},
  howpublished = {\url{http://snap.stanford.edu/data} (last accessed 2025)},
  year=  {2014},
}

@article{stehle2011high,
  title={High-resolution measurements of face-to-face contact patterns in a primary school},
  author={Stehl{\'e}, Juliette and Voirin, Nicolas and Barrat, Alain and Cattuto, Ciro and Isella, Lorenzo and Pinton, Jean-Fran{\c{c}}ois and Quaggiotto, Marco and Van den Broeck, Wouter and R{\'e}gis, Corinne and Lina, Bruno and Vanhems, Philippe},
  journal={PLoS ONE},
  volume={6},
  number={8},
  pages={e23176},
  year={2011}
}

@article{traag2019louvain,
  title={From Louvain to Leiden: guaranteeing well-connected communities},
  author={Traag, Vincent A and Waltman, Ludo and Van Eck, Nees Jan},
  journal={Scientific Reports},
  volume={9},
  number={1},
  pages={1--12},
  year={2019},
  publisher={Nature Publishing Group}
}

@Article{watts1998collective,
  title={Collective dynamics of ‘small-world’ networks},
  author={Duncan J. Watts and Steven H. Strogatz},
  journal={Nature},
  volume={393},
  number={6684},
  pages={440--442},
  year={1998},
  p2ublisher={Nature Publishing Group},
}

@Inproceedings{zhou2008preserving,
  title={Preserving privacy in social networks against neighborhood attacks},
  author={Zhou, Bin and Pei, Jian},
  booktitle={Proceedings of the IEEE International Conference on Data Engineering},
  pages={506--515},
  year={2008},

}

\section*{Supplementary information}\label{sec:supmat}
\section*{Measures and distances}
In the literature on $k$-anonymity for networks, various measures for anonymity have been introduced that reflect different attacker scenarios.
Supplementary Table~\ref{tab:measures} gives an overview of the most commonly used measures,~\cite{dejong2024comparison} a formal definition of the value for a given node $v$, the difference for two given nodes $(M(v) - M(w))$ and their accompanying time complexity.
In the difference computation for \textsc{degdist} and \textsc{vrq}, to ensure correct computation, we assume that nodes are traversed from highest to lowest degree. 
If either $V_{N(v)}$ or $V_{N(w)}$ contains more nodes, the sum of values of these remaining nodes is added to the difference.
When applied in domain-specific scenarios, the definitions in Supplementary Table~\ref{tab:measures} can be extended with weights to, for example, assign greater importance to differences in nodes than in edges.

\begin{table}[h!]
\small
\centering
\begin{tabular}{|r|l|l|l|l|}
\hline
\textbf{Measure}                 & $M(v)$                 & \textbf{Comp.}         & $M(v) - M(w)$                                                                                  & \textbf{Comp.}     \\ \hline
\textsc{degree}                        & $deg(v)$                   & $\mathcal{O}(1)$       & $|deg(v) -deg(w)|$                                                                                       & $\mathcal{O}(1)$   \\ \hline
\textsc{count}             & $|V_{N(v)}| + |E_{N(v)}|$     & $\mathcal{O}(|V|^2)$   & \begin{tabular}[c]{@{}l@{}}$| \ |V_{N(v)}| - |V_{N(w)}| \ | + |\ |E_{N(v)}| - |E_{N(w)}| \ |$\end{tabular} & $\mathcal{O}(1)$   \\ \hline
$(deg, tri)$             & $(deg, tri)$     & $\mathcal{O}(|V|^2)$   & \begin{tabular}[c]{@{}l@{}}$ |deg(v) - deg(w)| \ , \ |tri(v) - tri(w)| \ $\end{tabular} & $\mathcal{O}(1)$   \\ \hline
\textsc{degdist}           & $\sum_{v\in V_N} |E_{N(v)}|$ & $\mathcal{O}(|V|^2)$   & $\sum_{v, w \in V_{N(v)}, V_{N(w)}} |\ |E_{N(v)}| - |E_{N(w)}| \ |$                                                & $\mathcal{O}(|V|log(|V|))$ \\ \hline
\textsc{vrq}               & $\sum_{v\in V_N} degree(v)$   & $\mathcal{O}(|V|^2)$   & $\sum_{v, w\in V_{N(v)}, V_{N(w)}} |degree(v) - degree(w)|$                                                     & $\mathcal{O}(|V|log(|V|))$ \\ \hline
\textsc{$d$-$k$-Anonymity} & $C(v)$              & ? & GED(v, w)                                                                                             & NP-hard   \\ \hline
\end{tabular}
\caption{Measures for anonymity and definitions to compute the value for a given node $v$ (second column), the difference for two nodes (fourth column) and their corresponding time complexity (third and last columns). Time complexity is an open problem for \textsc{$d$-$k$-anonymity} as it requires isomorphism checks for which the complexity is unknown.
Here, $V_{N(v)}$ denotes the set of nodes in the neighborhood of node $v$, including node $v$ itself, and $E_{N(v)}$ denotes the set of edges in the neighborhood consisting of the edges to the neighbors of $v$ and the edges between neighbors of node $v$.
}
\label{tab:measures}
\end{table}

The first measure \textsc{degree} is equal to the node degree. 
This measure is simple to compute and models a weak attacker scenario. 
As a result, \textsc{degree} results in high anonymity in many networks. 
When using measures that assume more structural information, such as \textsc{count} or $(deg, tri)$, the variant considered in this paper, the anonymity decreases strongly. 

To show that \textsc{count} and $(deg, tri)$ are essentially the same measures in the non-fuzzy setting, the theorem below states that nodes equivalent according to \textsc{count} are also equivalent for $(deg, tri)$ and vice versa.

\begin{theorem}
    Given a graph $G=(V, E)$, nodes $v, w \in V$ are equivalent according to $(deg, tri)$ iff they are equivalent according to \textsc{count}. \label{thm:count}
\end{theorem}

\begin{proof}
Nodes are equivalent according to $(deg, tri)$ if 
$deg(v) = deg(w)$ and $tri(v) = tri(w)$,
and equivalent according to \textsc{count} if 
$|V_{N(v)}| = |V_{N(w)}|$ and $|E_{N(v)}| = |E_{N(w)}|$, i.e., they have the same counts of nodes and edges in their neighborhood.

For any node $v$ it holds that $deg(v) = |V_{N(v)}|-1$, as the $deg(v)$ equals the number of nodes to which $v$ connects. Hence, if two nodes have the same degree they have the same number of nodes in their neighborhood and the other way around.

Additionally, for any node $v$ it holds that $tri(v) = |E_{N(v)}| - deg(v)$, as $tri(v)$ equals the number of edges among neighbors, and $|E_{N(v)}|$ equals the number of edges from $v$ to its neighbors (which is equal to the node degree) and the number of edges among its neighbors.
Given that the degrees of the two nodes are equal and $tri(v)=tri(w)$, the number of edges in the neighborhood must also be equal for the equation to hold.
Similarly, if the degree and number of edges are equal, thus $deg(v)=deg(w)$ and $|E_{N(v)}| = |E_{N(w)}|$  the number of triangles must also be equal: $tri(v) = tri(w)$.
Hence, two nodes $v, w \in V$ are equivalent according to $(deg, tri)$ iff they are equivalent according to \textsc{count}. 
\end{proof}

For \textsc{$d$-$k$-anonymity}, the most complete measure in our set, the distance between two nodes equals the editing distances between node neighborhoods, i.e., the number of node/edge deletions/additions that needs to be applied to transform one neighborhood such that it is isomorphic to the other.
While \textsc{$d$-$k$-anonymity} is the most complete and intuitive to interpret, computing the graph edit distance is known to be NP-hard~\cite{blumenthal2020exact} which makes it computationally infeasible to use in this scenario for larger networks.
Other measures, \textsc{vrq} and \textsc{degdist} are, compared to $(deg, tri)$, computationally slightly more complex and more difficult to interpret.

To better understand the relation among the aforementioned measures, they can be ordered based on strictness.~\cite{dejong2024comparison}
If a measure $A$ is more strict than a measure $B$, then equivalence under measure $A$ requires equivalence under measure $B$, while equivalence under measure $B$ alone might be insufficient to guarantee equivalence under measure $A$.
For example, nodes equivalent under \textsc{$d$-$k$-anonymity} must also be equivalent under $(deg, tri)$-anonymity, while the converse does not hold.
As a result, to make a graph $k$-anonymous according to \textsc{$d$-$k$-anonymity}, it needs to be $k$-anonymous under $(deg, tri)$-anonymity.
Hence, the number of alterations required for $(deg, tri)$-anonymity is a lower bound, for the number of alterations required for $d$-$k$-anonymity.

\section*{($deg, tri)$-signatures and $\phi$-$k$-anonymity in real-world networks}
Supplementary Figure~\ref{fig:app:nm1} shows for the smallest 20 included real-world networks the degree and number of coinciding triangles of its nodes (which together form their $(deg, tri)$ signature) and for each node whether it is $\phi$-$2$-anonymous.

\begin{figure}[b!]
    \centering
    \includegraphics[width=0.9\textwidth]{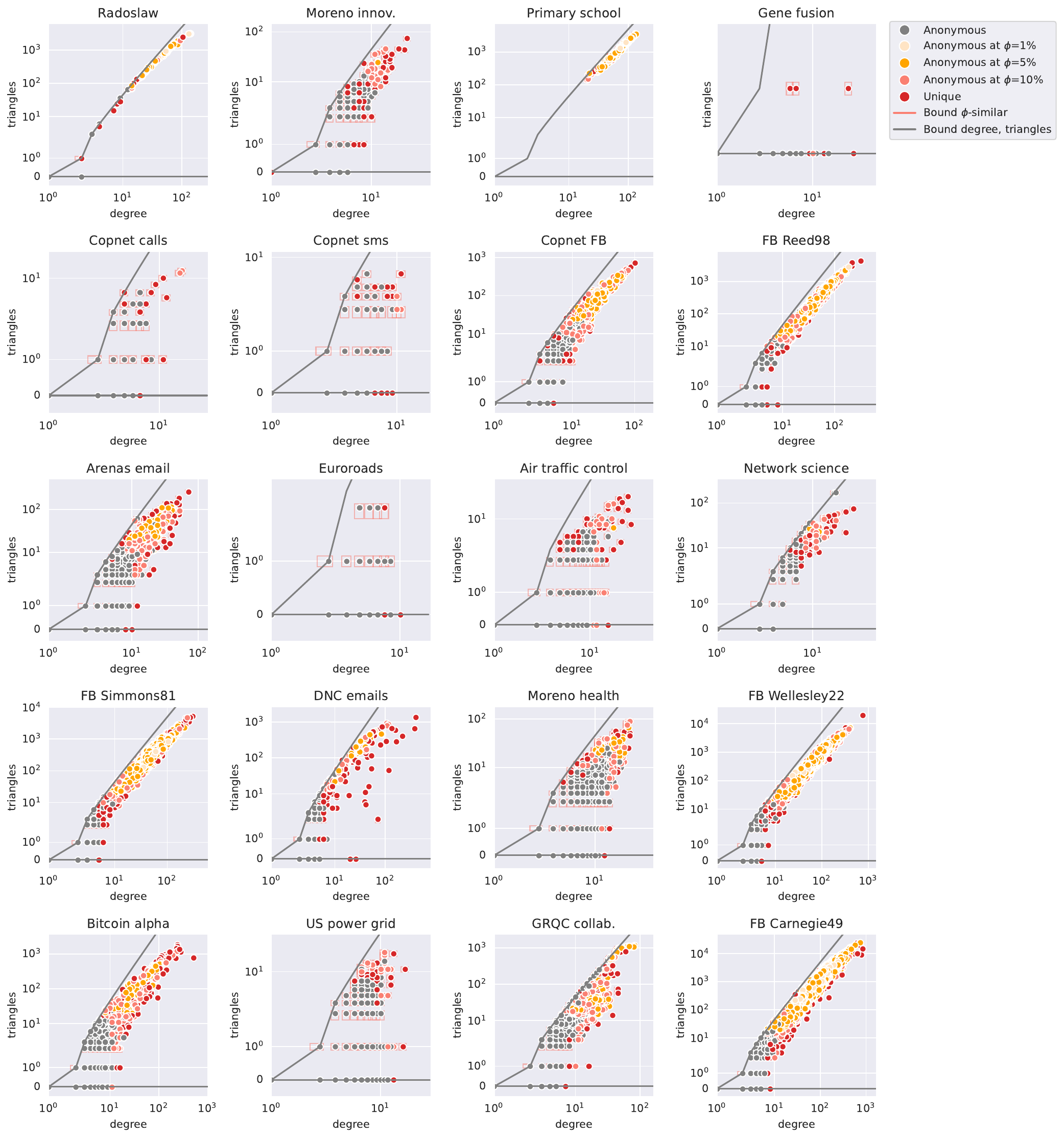}
    \caption{$\phi$-$k$-anonymity in 20 real-world networks. 
    In each subfigure the horizontal axis denotes the node degree, the vertical axis the number of triangles the node is part of which together form the $(deg, tri)$ node signature.
    The grey lines indicate the minimum and maximum number of triangles for each degree value, corresponding to the number of triangles in a star graph and a complete graph respectively. 
    Each dot represents the degree and triangles value for at least one node. 
    Grey nodes are 2-anonymous, red nodes unique. Yellow, orange and pink nodes are 2-$\phi$-anonymous for the $\phi$ value indicated by color. 
    Rectangles indicate the area to which each center nodes is $\phi$-similar with $\phi=10\%$.
    }
    \label{fig:app:nm1}
\end{figure}

\clearpage

\section*{Graph models with 1,000 nodes}
To show how results scale to a larger number of nodes in graph models, this section contains results for each of the three graph models with 1,000 nodes.
Supplementary Figure~\ref{app:fig:model:anon} contains results on anonymity with different levels of uncertainty $\phi$, and Supplementary Fig.~\ref{app:fig:model:anonymization} on anonymization.
Compared to results in Fig. 3 in the main text of this paper, the transition where most nodes in the graph models become anonymous takes place with a smaller value $\phi$.
These results show that uncertainty is also effective, if not more effective, when considered for larger graphs.

\begin{figure}[h]
    \centering
    \includegraphics[width=0.9\linewidth]{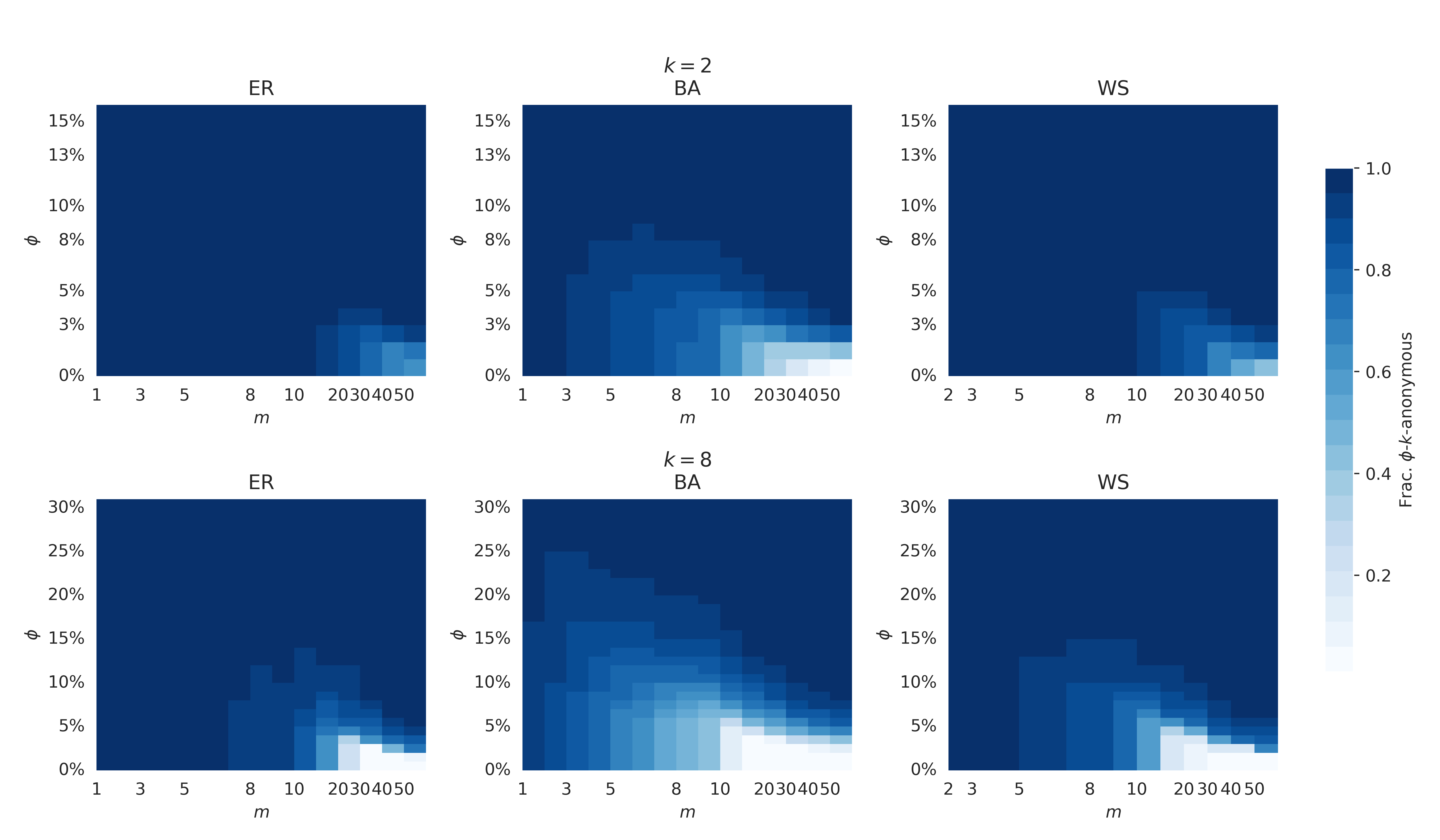}
    \caption{$\phi$-$k$-Anonymity for $k=2$ (top) and $k=8$ (bottom) in Erdős–Rényi (ER, left) Barabási–Albert (BA, middle) and Watts Strogatz (WS, right) graph models with 1,000 nodes.
    The horizontal axis denotes the number of edges per node $m$, the vertical axis the uncertainty $\phi$.
    Color indicates the fraction of $\phi$-$k$-anonymous nodes ranging from white (all nodes unique) to dark blue (all nodes are anonymous).}
    \label{app:fig:model:anon}
\end{figure}

\begin{figure}[h]
    \centering
    \includegraphics[width=0.9\linewidth]{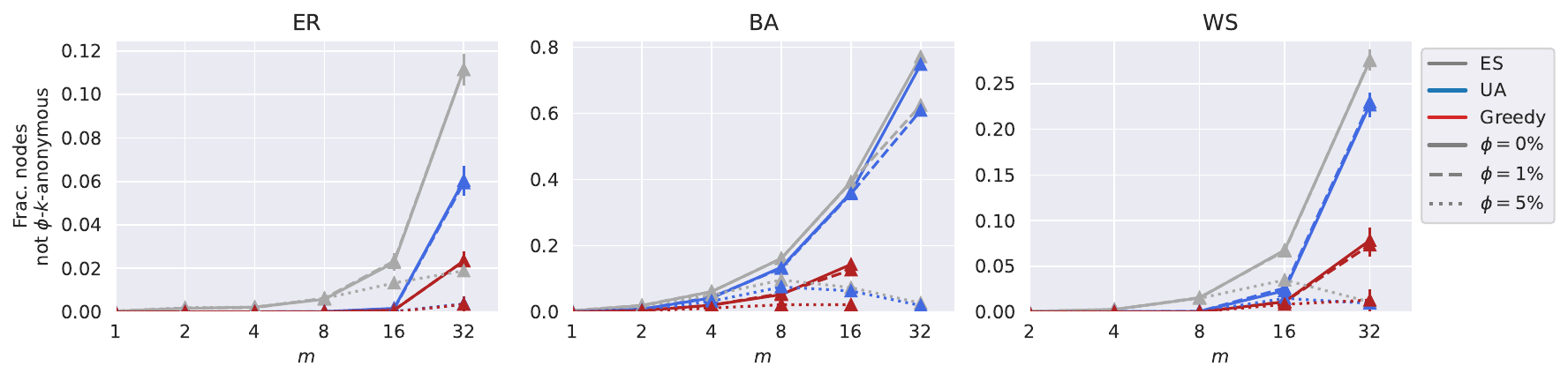}
    \caption{$\phi$-$k$-Anonymity after budgeted anonymization (deleting 5\% of the edges) on graph models ER (left), BA (middle) and WS (right) with 1,000 nodes using different anonymization algorithms (color) and levels of uncertainty $\phi$ (linestyle). The horizontal axis denotes the number of connections added to each node, and the vertical axis denotes the fraction of unique nodes for $k$=2 and varying $\phi$. Vertical lines indicate the standard deviation.}
    \label{app:fig:model:anonymization}
\end{figure}

\newpage
\section*{($deg, tri)$-signatures and $\phi$-$k$-anonymity in real-world networks after anonymization}
Supplementary Figure~\ref{app:fig:nmanon} shows for the 12 network datasets used in anonymization experiments, the $(deg, tri)$-signatures of its nodes, and for each node whether it is $\phi$-$2$-anonymous.
\begin{figure}[h]
    \centering
    \includegraphics[width=\linewidth]{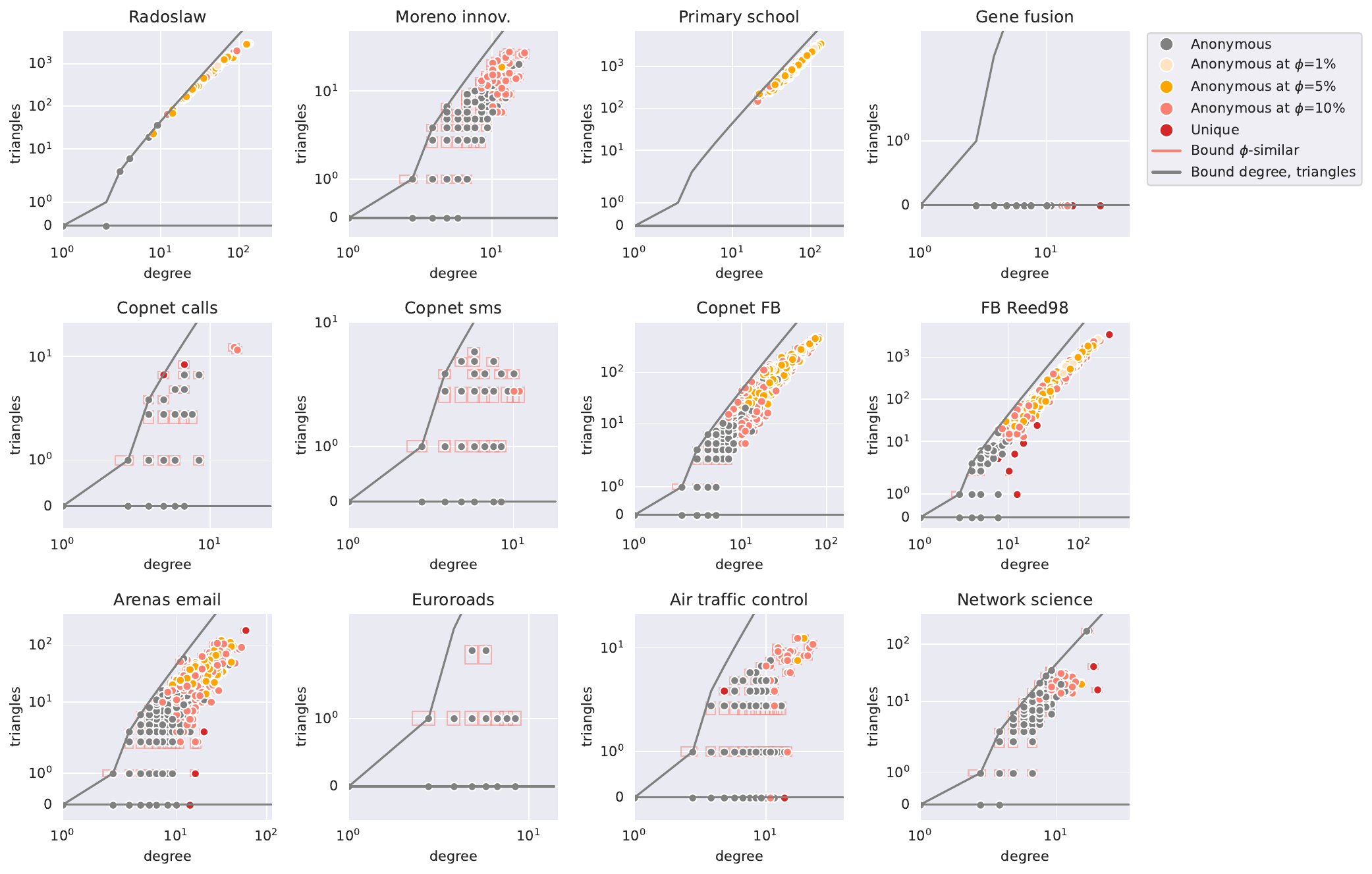}
    \caption{
    $\phi$-$k$-anonymity in 12 real-world networks after \textsc{greedy} anonymization. 
    In each subfigure the horizontal axis denotes the node degree, the vertical axis the number of triangles the node is part of which together form the $(deg, tri)$ node signature.
    The grey lines indicate the minimum and maximum number of triangles for each degree value, corresponding to the number of triangles in a star graph and a complete graph respectively. 
    Each dot represents the degree and triangles value for least one node after anonymization. 
    Grey nodes are 2-anonymous, red nodes unique. Yellow, orange and pink nodes are 2-$\phi$-anonymous for the $\phi$ value indicated by color. 
    Rectangles indicate the area to which each center nodes is $\phi$-similar with $\phi=10\%$.
    }
    \label{app:fig:nmanon}
\end{figure}

\newpage
\section*{Utility}
Supplementary Figures~\ref{fig:util:es}, \ref{fig:util:ua} and \ref{fig:util:greedy} show for the 12 networks used in the anonymization experiments how well utility is preserved in the networks after applying the corresponding anonymization algorithm.

\begin{figure}[b!]
    \centering
    \includegraphics[height=0.8\textheight]{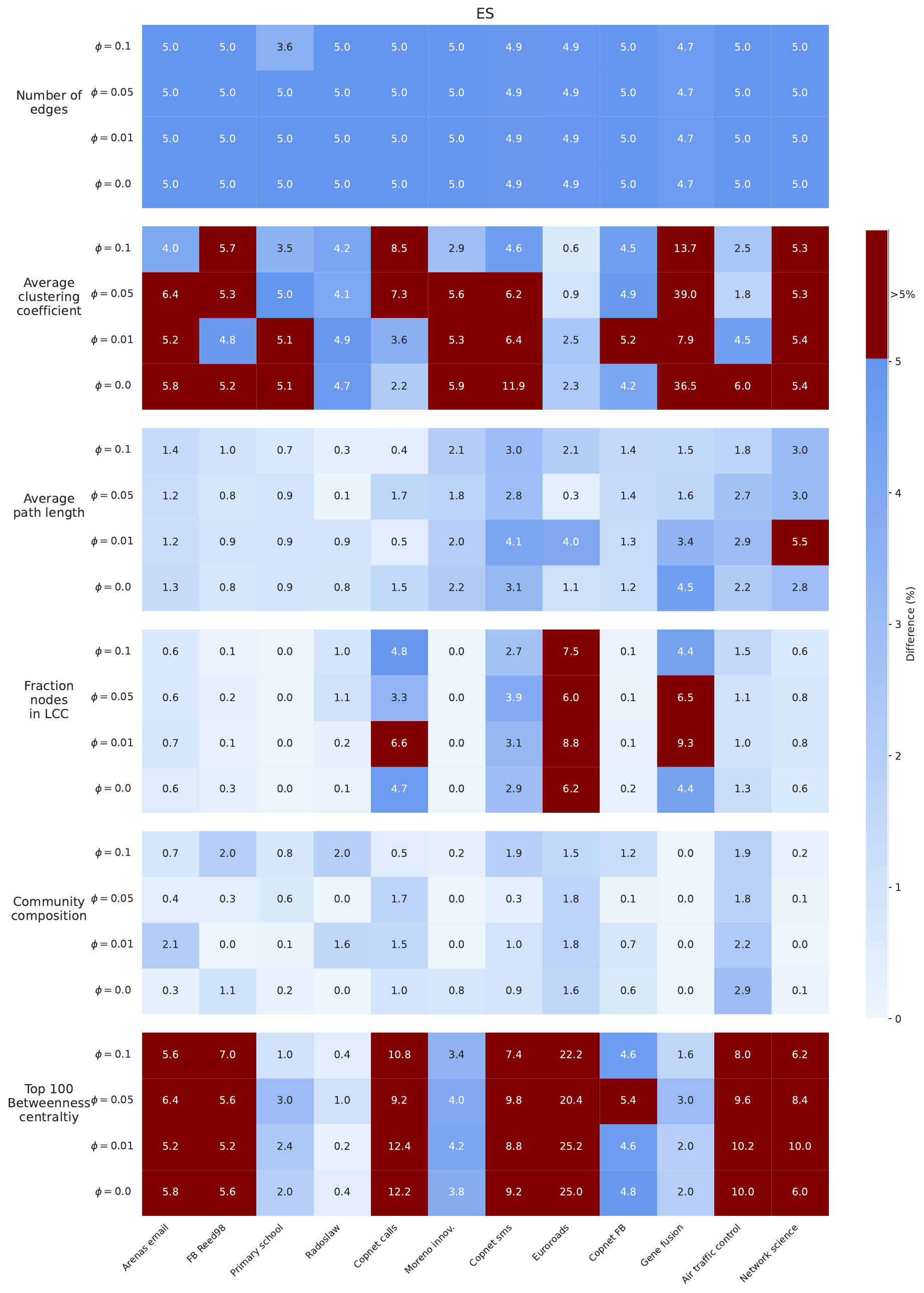}
    \caption{Data utility after budgeted anonymization with \textsc{es}. The plot in each row shows the relative difference,
    comparing the network before and after anonymization, in a data utility metric. Results are shown for 12 networks (horizontal axis), and five levels of uncertainty $\phi$ (vertical axis). White indicates no difference, while the darker shades of blue indicate larger differences. Red cells indicate differences larger than 5\%.}
    \label{fig:util:es}
\end{figure}

\begin{figure}
    \centering
    \includegraphics[height=0.8\textheight]{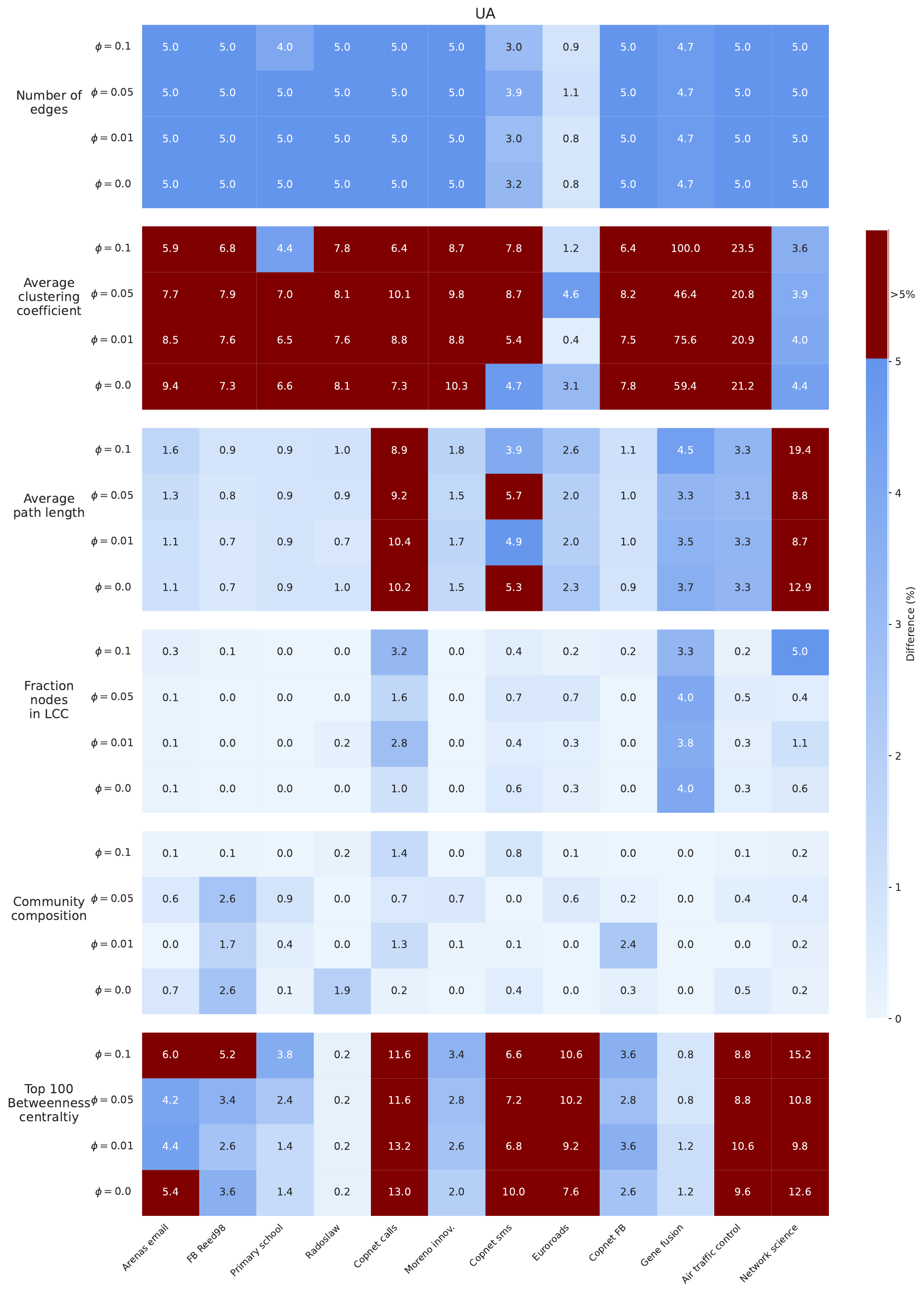}
    \caption{Data utility after budgeted anonymization with \textsc{ua}. The plot in each row shows the relative difference, comparing the network before and after anonymization, in a data utility metric. Results are shown for 12 networks (horizontal axis), and five levels of uncertainty $\phi$ (vertical axis). White indicates no difference, while the darker shades of blue indicate larger differences. Red cells indicate differences larger than 5\%.}
    \label{fig:util:ua}
\end{figure}

\begin{figure}
    \centering
    \includegraphics[height=0.8\textheight]{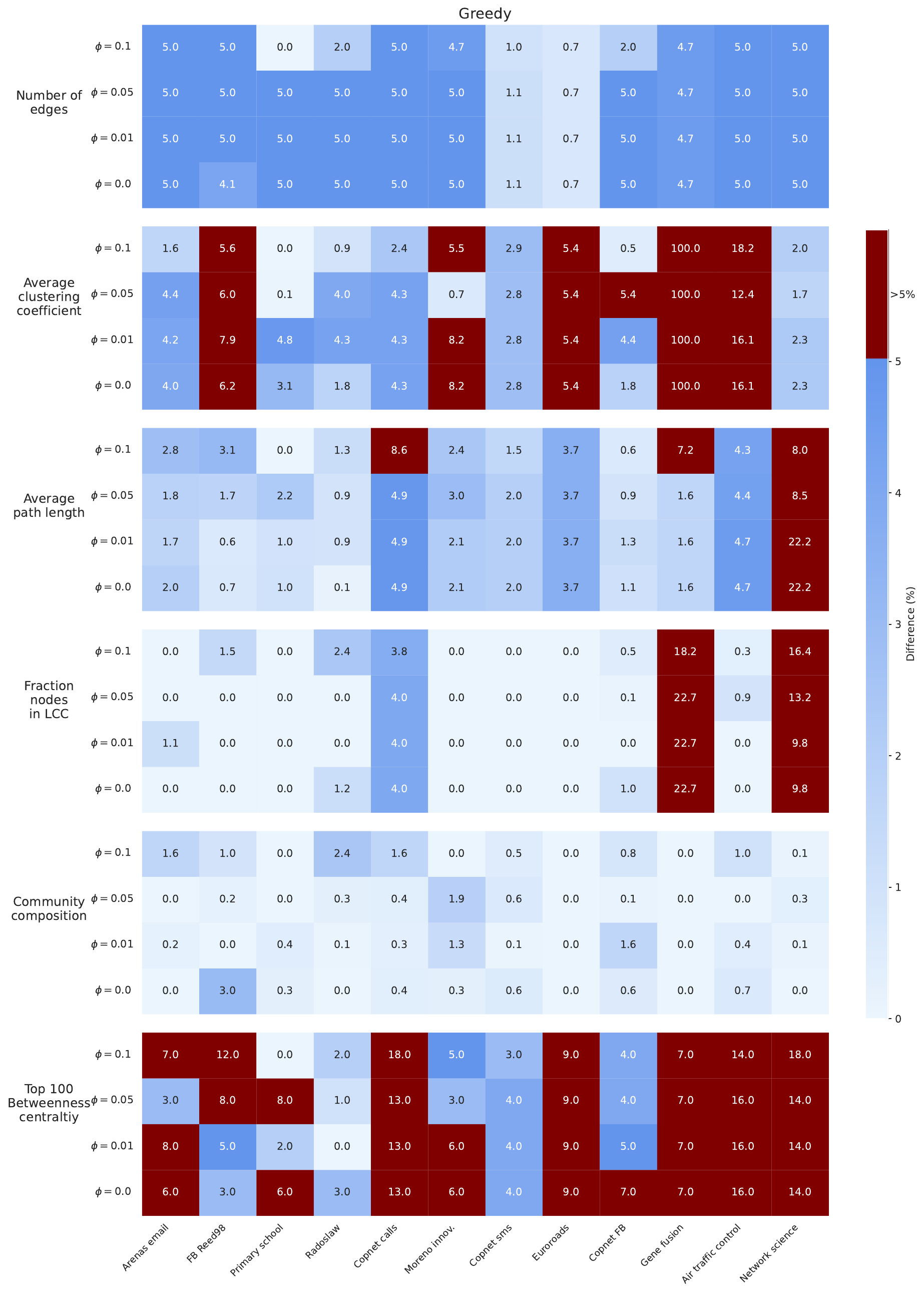}
    \caption{Data utility after budgeted anonymization with \textsc{greedy}. The plot in each row shows the relative difference, comparing the network before and after anonymization, in a data utility metric. Results are shown for 12 networks (horizontal axis), and five levels of uncertainty $\phi$ (vertical axis). White indicates no difference, while the darker shades of blue indicate larger differences. Red cells indicate differences larger than 5\%.}
    \label{fig:util:greedy}
\end{figure}

\section*{Measuring similarity in community structure with NMI}
Supplementary Figure~\ref{fig:util:comm} shows, for the 12 networks included in the anonymization experiments, the obtained values for $NMI_{stability}$ and $NMI_{anon}$ for the community partitions found by the \textsc{leiden}~\cite{traag2019louvain} algorithm after applying \textsc{greedy} anonymization. 
The results show that there are instances where $NMI_{anon}$ exceeds $NMI_{stability}$.
This implies that, on average, the community partitions found for the anonymized network ($C'$) have more in common with those found for the original network ($C$) than independently detected community partitions of $C$ agree with each other.
Cases where $NMI_{anon} > NMI_{stability}$ are more frequent for networks with unstable communities, as indicated by low $NMI_{stability}$ values.
In most of these unstable instances $NMI_{anon}$ exceeds $NMI_{stability}$ by less than 0.025.

\begin{figure}[h]
    \centering
    \includegraphics[width=0.5\textwidth]{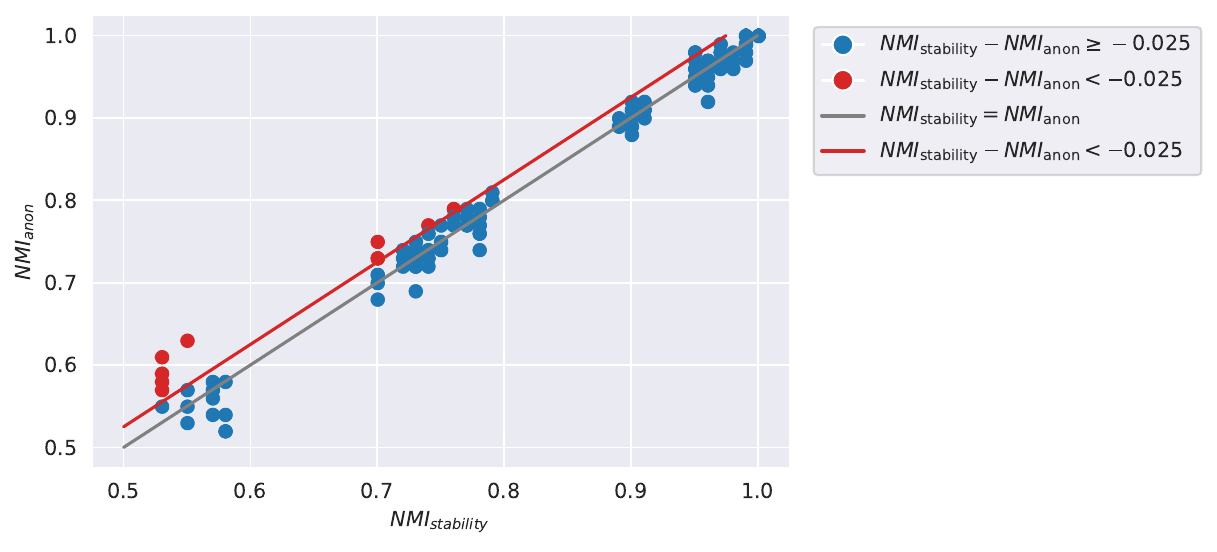}
    \caption{Values for $NMI_{stability}$ and $NMI_{anon}$ observed for real-world networks after applying \textsc{greedy} anonymization. Red dots indicate that $NMI_{anon}$ is 0.025 higher than $NMI_{stability}$. Blue dots indicate that $NMI_{anon}$ is less than 0.025 higher than $NMI_{stability}$.}
    \label{fig:util:comm}
\end{figure}

\end{document}